\documentclass[11pt,a4paper]{article}

\usepackage{amssymb, amsmath, amsthm, amsfonts}
\usepackage{graphicx}
\usepackage{dcolumn}
\usepackage{bm}
\usepackage{mathrsfs}
\usepackage[table]{}
\usepackage[toc,page]{appendix}
\usepackage{subfig}
\usepackage{authblk}
\usepackage[T1]{fontenc}
\usepackage{multirow}
\usepackage{booktabs}
\usepackage[top=1.0in, bottom=1in, left=1.0in, right=1.0in]{geometry}

\newtheorem{theorem}{Theorem}[section]

\newtheorem{prop}[theorem]{Proposition}

\newtheorem{corollary}[theorem]{Corollary}

\newtheorem{definition}[theorem]{Definition}

\begin{document}

\title{Gyrokinetic Vlasov-Poisson model derived by hybrid-coordinate transform of the distribution function}

\author[1]{Shuangxi Zhang \thanks{zshuangxi@gmail.com}}
\affil[1]{IRMA, Universit\'e de Strasbourg, France \& Inria TONUS team}
\affil[1]{University of Science and Technology of China}

\date{\today}


\maketitle

\begin{abstract}
This paper points out that the full-orbit density obtained in the standard electrostatic gyrokinetic model is not truly accurate at the order $\varepsilon^{\sigma-1}$ with respect to the equilibrium distribution $e^{-\alpha \mu}$ with $\mu \in (0, \mu_{\max})$, where $\varepsilon$ is the order of the normalized Larmor radius, $\varepsilon^{\sigma}$ the order of the amplitude of the normalized electrostatic potential, and $\alpha$ a factor of $O(1)$. 	
This error makes the exact order of the full-orbit density not consistent with that of the approximation of the full-orbit distribution function. By implementing a hybrid coordinate frame to get the full-orbit distribution, specifically, by replacing the magnetic moment on the full-orbit coordinate frame with the one on the gyrocenter coordinate frame to derive the full-orbit distribution transformed from the gyrocenter distribution, it's proved that the full-orbit density can be approximated with the exact order being $\varepsilon^{\sigma-1}$. 
The numerical comparison  between the new gyrokinetic model and the standard one was carried out using Selalib code for an initial distribution proportional to $\exp(\frac{-\mu B}{T_i})$ in constant cylindrical magnetic field configuration with the existence of electrostatic perturbations. In such a configuration, the simulation results exhibit similar performance of the two models.
\end{abstract}

\section{Introduction}\label{introduction}

The strong magnetic field provides a potential mean to create an environment to confine the hot plasma ionized from light elements such as Hydrogen, Tritium and Deuterium, to achieve the fusion purpose by collisions\cite{0741-3335-44-8-701, Kadomtsev1966, wesson2004}. While the experiments of the magnetized plasma is significant, the numerical simulation provides another approach to predict the behaviour of the plasma\cite{Birdsall1985, Buchne2003, Lee1987}. One important objective for the prediction is the low-frequency electrostatic turbulence, which is recognized as the factor to contribute to the plasma anomalous transport\cite{Garbet2010, wesson2004, Biglari1990, Horton1999, Diamond2005}. So far, the gyrokinetic simulation based on the standard gyrokinetic model (SGM)\cite{Dubin1983, Lee1987, Garbet2010} is widely conceived as a strong tool to predict the behaviour  of those low-frequency turbulence\cite{Lee1987, Jenko2000, Candy2003, Grandgirard2006, Latu2017, Terry2017, Idomura2008, Chen2003}, since it reduces the 6D Vlasov equation to a 5D one with the magnetic moment being constant and keeps the kinetic effects\cite{Dubin1983,Brizard1990,Hahm1988, Frieman1982, Cary1983, Littlejohn1983, Brizard2007, Cary2009, Brizard1989, Sugama2000}. A simple derivation of the electrostatic standard gyrokinetic model is given in Appendix \ref{coordinates}.

The gyrokinetic simulations implement the gyrokinetic Vlasov equation  to compute the evolution of the gyrocenter distribution, which is totally defined on the gyrocenter coordinate frame\cite{Garbet2010} with the initial gyrocenter distribution given at the beginning of the simulation. To simulate a realistic magnetized plasma, the gyrocenter distribution of the magnetic moment  $\mu$ is usually chosen as $\exp(-\alpha\mu)$ with $\alpha\equiv \frac{B}{T_i}$ and $\mu \in (0,\mu_{\max})$,  and ideally, $\mu$ should belong to the domain $(0,+\infty)$. The definition of $\mu$ and other notations  used in the following explanations can be found in  Sec.(\ref{notation}). Meanwhile,  due to that the Coulomb force happens on the full-orbit coordinate frame, the electrostatic potential is computed by the quasi-neutrality equation (QNE) defined on the full-orbit coordinate frame and as a simplified version of Poisson equation\cite{Garbet2010, Brizard1990}. 

Before going on to the next explanation, we need the definition of the ``exact order'' and ``uncertain order''. 
\begin{definition}
The ``exact order'' in this paper denotes the highest order at which the associated quantity is exactly right as the result of the approximation imposed on this quantity, while the ``uncertain order'' denotes the lowest order at which the associated quantity is ignored. 
\end{definition}

In this paper, the electrostatic potential is normalized by $B_0L_0 v_t$.  The order of the amplitude of electrostatic potential $\phi$ is extracted so that the electrostatic potential is written as $\varepsilon^{\sigma}\phi$, where  $O(|\phi|)=O(1)$ and $\varepsilon^{\sigma}$ is the order of the potential with $\varepsilon\equiv \frac{{m{v_t}}}{{q{B_0 L_0}}}$ and $\sigma$ an exponent independent of $\varepsilon$ used to signifying the order of the amplitude of the potential. The meanings of all the symbols used here can be found in Subsec.(\ref{normalorder}).
Ref.(\cite{Hahm1988}) gives the order $O(\frac{e\phi}{T_i})=O(\varepsilon)$, which can be translated into $\sigma=2$ in terms of the normalization scheme used in this paper. Eq.(\ref{c17}) in Appendix (\ref{generator}) points out that $\sigma< 3$ should be satisfied to make sure that the electrostatic potential term $\varepsilon^{\sigma} \phi$ is the exact term contained by the orbit equation. So in this paper the reasonable region of $\sigma$ is chosen as $2\le \sigma <3$.

Due to that the exact order of the approximation to get $f(\mathbf{z})$ of SGM in Eq.(\ref{qq5}) is $\varepsilon^{\sigma-1}$, it's a natural idea that a density of exact order $\varepsilon^{\sigma-1}$ could be derived by $\int {f(\mathbf{z})Bd\mu_1du_1 d\theta_1}$, so that QNE would be of the exact order $\varepsilon^{\sigma-1}$. However, because it's difficult to compute the lower bound $\mu_{1\min}(\mathbf{x},\theta_1)$ of the domain of $\mu_1$ which is mapped from the domain of $\mu$ as shown in Subsec.(\ref{error.1}),  
the standard model treats the domain of $\mu_1$ in the full-orbit coordinate frame the same with that of $\mu$ being $(0,+\infty)$ in the gyrocenter coordinate frame.  As proved in Sec.(\ref{error}), for the distribution $\exp(-\alpha \mu)$ of $\mu$ which is usually used for a realistic plasma,  this treatment leads to an error of order $O(\varepsilon^{\sigma-1})$ to the full-orbit density. Therefore, the exact order of SGM is not  $O(\varepsilon^{\sigma-1})$.  Eventually, the error of the order $O(\varepsilon^{\sigma-1})$ produced by computing $n(\mathbf{x})$ is inherited by QNE. 

In this paper, instead of $\mathbf{z}=(\mathbf{x},\mu_1,u_1,\theta_1)$, which is the full-orbit coordinates with the velocity written in cylindrical coordinates as shown in Subsec.(\ref{new.1}), the hybrid coordinates $(\mathbf{x},\mu,u_1,\theta_1)$ is implemented to obtain the distribution on the full-orbit coordinate frame, so that the domain of $\mu$ can be safely used. The functional relationship between $\mu_1$ and $\mu$ is given by Subsec.(\ref{error.1}). With this hybrid coordinates frame, it's proved in Sec.(\ref{new}) that the density and QNE can be derived with the exact order being $O(\varepsilon^{\sigma-1})$.
The numerical comparison is carried out between the new model and the standard one based on the SELALIB platform\cite{S}.  The rest of the paper is arranged as follows. Sec.(\ref{notation}) introduces the basic scales and their respective orders, as well as the notations which are used in the context. Sec.(\ref{error}) presents the proof that the exact order of the full-orbit density derived by SGM is not $O(\varepsilon^{\sigma-1})$. The hybrid coordinate transform and the proof that the exact order of the new full-orbit density is $O(\varepsilon^{\sigma-1})$ are given in Sec.(\ref{new}).  Sec.(\ref{model}) lists the normalized new gyrokinetic model and SGM. The various algorithms, the parallelization scheme, as well as the numerical results are presented in Sec.(\ref{numerical}).

\section{The notations and the basic orders}\label{notation}

\subsection{The coordinate transforms used in gyrokinetic theory and the metrics}\label{new.1}

The procedure to derive the gyrokinetic model is composited by two parts. The first one is to derive the coordinate transform by decoupling  the gyroangle from the dynamics of other coordinates, while the second one is to obtain the gyrokinetic quasi-neutral equation by inducing the transformation of the distribution through the derived coordinate transforms\cite{Dubin1983,Garbet2010,Brizard1990}. Generally, four kinds of coordinate frameworks are involved in the procedure. The first one is the full-orbit coordinate with the velocity part in Cartesian coordinates. It's denoted as $\bar{\mathbf{z}}\equiv (\mathbf{x},\mathbf{v})$ here. The second one is obtained by transforming $\bf{v}$ into the cylindrical coordinates, and it's written as $\mathbf{z}\equiv (\mathbf{x}, \mu_1, u_1, \theta_1)$ with $\mu_1\equiv \frac{m\mu_1^2}{2B(\bf{x})}$. The $\bf{x}$ component in $\bf{z}$ is still in full-orbit frame. The third one is the guiding-center coordinates $\bar{\mathbf{Z}}=(\bar{\mathbf{X}},\bar{\mu},\bar{U},\bar{\theta})$, which is derived by decoupling $\bar{\theta}$ from the dynamics of the other coordinate components without the existence of the perturbation. The fourth one is the gyrocenter coordinate $\mathbf{Z}=(\mathbf{X},\mu,U,\theta)$ which is derived by decoupling $\bar{\theta}$ from the dynamics of the other coordinate components with the existence of the perturbation.  The coordinate transforms between $\bar{\mathbf{z}}$, $\mathbf{z}$,$\bar{\mathbf{Z}}$ and $\mathbf{Z}$ are denoted as $\psi_{f}:\bar{\mathbf{z}}\to {\mathbf{z}}$, $\psi_{gc}:\mathbf{z}\to \bar{\mathbf{Z}}$ and $\psi_{gy}:\bar{\mathbf{Z}}\to \mathbf{Z}$, respectively, while the distributions on the four kinds of coordinates are written as $\bar{f}(\bar{\bf{z}})$, $f({\bf{z}})$, $\bar{F}(\bar{\bf{Z}})$ and $F({\bf{Z}})$, respectively. The coordinate transform $\psi_{gc}$ and $\psi_{gy}$ is realised by the Lie transform perturbative method for a noncanonical system. A simple introduction of this method is given by Appendix. \ref{lie} and the details can be found in Ref.(\cite{Cary1983}). The details of the derivation of the coordinate transforms are given in Appendix.(\ref{coordinates})

The functional relationship between the distributions are listed below
\begin{subequations}\label{}
\begin{eqnarray}
f(\mathbf{z})&=&\bar{f}(\psi_{f}^{-1}({\mathbf{z}})),  \nonumber \\
\bar{F}(\bar{\mathbf{Z}})&=&\bar{f}(\psi_{f}^{-1}\psi_{gc}^{-1}(\bar{\mathbf{Z}})), \nonumber \\
F(\mathbf{Z})&=&\bar{f}(\psi_{f}^{-1}\psi_{gc}^{-1}\psi_{gy}^{-1}({\mathbf{Z}})). \nonumber 
\end{eqnarray}
\end{subequations}
$\bar{f}(\bar{\mathbf{z}})$ satisfies the Vlasov equation $\frac{d\bar{f}(\bar{\mathbf{z}})}{dt}=0$,
where the symbol $d$ denotes the full derivative.  This Vlasov equation induces other Vlasov equations for $f(\mathbf{z})$, $\bar{F}(\bar{\mathbf{Z}})$ and $F(\mathbf{Z})$ and they can be uniformly written as 
$$\left (\frac{\partial }{\partial t}+\frac{d\mathbf{z}_i}{dt}\cdot \frac{\partial }{\partial \mathbf{z}_i}\right) f_i(\mathbf{z}_i)=0, $$
where $\mathbf{z}_i$ with $i=1,2,3,4$ denote $\bar{\mathbf{z}}, \mathbf{z}, \bar{\mathbf{Z}}, \mathbf{Z}$, respectively, while $f_i$s denote their respective distributions.

The total number is derived by integrating the distributions on their respective phase space 
$$\int f_i(\mathbf{z}_i)\eta_i d^6 \mathbf{z}_i.$$
Here, $\eta_i$ is the determinant of the metric of the respective phase space. Due to the conservation of the total number, the determinant of the metrics can be obtained as:
\begin{subequations}\label{}
\begin{eqnarray}
\eta_1(\bar{\mathbf{z}})&=&1, \nonumber \\
\eta_2(\mathbf{z})&=&\left | {\frac{{{d^6}\psi _f^{ - 1}\left( {\bf{z}} \right)}}{{{d^6}{\bf{z}}}}} \right |=\frac{B(\mathbf{x})}{m}, \nonumber \\
\eta_3(\bar{\mathbf{Z}})&=& \left | {\frac{{{d^6}\psi _f^{ - 1}\psi _{gc}^{ - 1}\left(\bar{\mathbf{Z}} \right)}} {{{d^6} \bar{\mathbf{Z}}}}} \right |, \nonumber \\
\eta_4(\mathbf{Z})&=& \left |{\frac{{{d^6}\psi _f^{ - 1}\psi _{gc}^{ - 1}\psi _{gy}^{ - 1}\left( {\bf{Z}} \right)}}  {{{d^6}{\bf{Z}}}}} \right |. \nonumber
\end{eqnarray}
\end{subequations}

$\eta_2(\mathbf{z})$ will be repeatedly used in the paper to get the density on the particle-coordinate spatial space.

\subsection{The equilibrium distribution}\label{distri}

The gyrokinetic Vlasov simulation implements an initial distribution on the gyrocenter coordinates frame\cite{Garbet2010}. The equilibrium distribution $F_{s0}(\mathbf{X},\mu,U)$ for charged particles with the species denoted by the subscript ``s'' can be decomposed as the product between  the parallel part and the perpendicular part
\begin{equation}\label{qq22} 
{F_{s0}}({\bf{X}},{\mu},{U}) = {n_0}({\bf{X}}){F_{s0\parallel }}\left( {{\bf{X}},{U}} \right){F_{s0 \bot }}\left( {{\bf{X}},{\mu}} \right),
\end{equation}
with the probability conservation being satisfied by
\begin{subequations}\label{qq30}
\begin{eqnarray}
&&\int {{F_{s0\parallel }}d{U}}  = 1, \\
&&\int {{F_{s0 \bot }}B({\bf{X}})d{\mu }d{\theta} } =1,
\end{eqnarray}
\end{subequations}
where  $B({\bf{X}})$ as the amplitude of the equilibrium magnetic field plays the role of Jacobian.   As usual, the equilibrium perpendicular distribution \cite{Jenko2000, Candy2003, Grandgirard2006}
\begin{equation}\label{e65}
 {F_{s0 \bot }}=\frac{m_s}{2\pi T_s}\exp(\frac{-\mu B}{T_s})
 \end{equation}
 is chosen in this paper.

\subsection{The nondimensionalization and the basic orders}\label{normalorder}

Gyrokinetic theory begins with implementing Lie transform perturbative theory  on the fundamental one-form to find out the coordinate transform. The orders of the length scale and amplitude of the equilibrium and perturbative quantities are firstly involved at this step and the exact and uncertain orders are inherited by the next procedure. So the fundamental one-form and the basic orders are first given here.   

\subsubsection{The nondimensionalization of quantities by nondimensionalizing the fundamental Lagrangian one-form }
The Lagrangian differential 1-form which determines the orbit of a test charged particle in magnetized plasmas \cite{Brizard1990, Dubin1983, Littlejohn1983, Littlejohn1982, Cary1983} is
\begin{equation}\label{a119}
\gamma = \left( {q{\bf{A}}\left( {{{\bf{x}}}} \right) + m{\bf{v}}} \right)\cdot d{\bf{x}} - (\frac{1}{2}m{v^2}+q \phi(\mathbf{x},t))dt.
\end{equation}
$(\mathbf{x},\mathbf{v})$ is the full particle coordinate frame.
The test particle is chosen from a thermal equilibrium plasma ensemble, e.g., the thermal equilibrium plasma in tokamak. Therefore, $\mathbf{A},\mathbf{v},\mathbf{x},t,\mathbf{B}, \phi,\mu$ can be nondimensionalized by $A_0\equiv B_0 L_0,v_t,L_0,L_0/v_t,B_0,A_0 v_t, mv^2_t/B_0$, respectively. $B_0, L_0$ are the characteristic amplitude and spatial length of the magnetic field, respectively. $v_t$ is the thermal velocity of the particle ensemble which contains the test particle.

The detailed normalization procedure of $\gamma$ is given as follows. First, both sides of Eq.(\ref{a119}) are divided by $m{v_t}{L_0}$. The first term of RHS of Eq.(\ref{a119}) is $\frac{{q{A_0}}}{{m{v_t}}}\frac{{{\bf{A}}\left( {{{\bf{x}}}} \right)}}{{{A_0}}}\cdot\frac{{d{{\bf{x}}}}}{{{L_0}}}$, which is further written as $\frac{1}{\varepsilon }{\bf{A}}\left( {{{\bf{x}}}} \right)\cdot d{{\bf{x}}}$, with the replacement:  $\frac{{{\bf{A}}\left( {{{\bf{x}}}} \right)}}{{{A_0}}} \to {\bf{A}}\left( {{{\bf{x}}}} \right),\frac{{d{{\bf{x}}}}}{{{L_0}}} \to d{{\bf{x}}}$ and 
\begin{equation}\label{e60}
\varepsilon\equiv \frac{{m{v_t}}}{{q{B_0 L_0}}}=\frac{\rho}{L_0} , \rho\equiv \frac{{m{v_t}}}{{q{B_0}}}. \nonumber 
\end{equation}
Other terms can be nondimensionalized in the same way. For the convenience of the ordering analysis, the order of the dimensionless quantity $|\phi|$ is extracted as an independent parameter and is denoted as $\varepsilon^\sigma$ based on the parameter $\varepsilon$, where $\sigma$ is an exponential index independent of $\varepsilon$. Alternatively, 
\begin{equation}\label{e62}
\phi\to \{\varepsilon^\sigma \phi, \;\; O(|\phi|)=O(1)\}. \nonumber 
\end{equation}

Eventually, we could derive a normalized Lagrangian 1-form 
\begin{equation}\nonumber
\frac{{{\gamma}}}{{m{v_t}{L_0}}} = \left( {\frac{1}{\varepsilon }{\bf{A}}\left( {\bf{x}} \right) + {\bf{v}}} \right)\cdot d{\bf{x}} - (\frac{1}{2}{{\bf{v}}^2} + \frac{\varepsilon^\sigma}{\varepsilon }\phi \left( {{\bf{x}},t} \right))dt,
\end{equation}
Now, multiplying both sides by $\varepsilon$, and rewriting $\frac{{{\varepsilon \gamma}}}{{m{v_t}{L_0}}}$ to be $\gamma$, the normalized 1-form becomes
\begin{eqnarray}\label{g2}
\gamma  = \left( {{\bf{A}}\left( {\bf{x}} \right) + \varepsilon {\bf{v}}} \right)\cdot d{\bf{x}} - \left( {\varepsilon \frac{{{v^2}}}{2} + \varepsilon^\sigma \phi \left( {{\bf{x}},t} \right)} \right)dt.
\end{eqnarray}
Since a constant factor $\frac{\varepsilon }{{m{v_t}{L_0}}}$ doesn't change the dynamics determined by the Lagrangian 1-form, the Lagrangian 1-form given by Eq.(\ref{g2}) possesses the same dynamics with that given by Eq.(\ref{a119}).

The velocity can be written in cylindrical coordinates, by transforming $(\mathbf{x},\mathbf{v})$ to $(\mathbf{x},u_1,\mu_1,\theta_1)$, where $u_1$ is parallel velocity and $\mu_1$ is magnetic moment, with their definitions being $u_1\equiv \mathbf{v}\cdot \mathbf{b}$ and $\mu_1\equiv v_\perp^2/2B(\mathbf{x})$. The unit vector of the perpendicular velocity is 
$${\widehat {\bf{v}}_ \bot } \equiv \left( {{{\bf{e}}_1}\sin \theta  + {{\bf{e}}_2}\cos \theta } \right).$$
 $(\mathbf{e}_1,\mathbf{e}_2,\mathbf{b})$ are orthogonal mutually and $\mathbf{b}$ is the unit vector of the equilibrium magnetic field. After this transformation, $\gamma$ becomes
\begin{eqnarray}\label{a11}
\gamma=\gamma_0+ \varepsilon \gamma_1 +\varepsilon^\sigma \gamma_\sigma
\end{eqnarray}
which can be splitted into three parts as
\begin{subequations}\label{a12}
\begin{eqnarray}
{\gamma _0} &=& {{\bf{A}}}\left( {\bf{x}} \right)\cdot d{\bf{x}}, \\
\varepsilon {\gamma _1} &=& \varepsilon\left( {{u_1}{\bf{b}} + \sqrt {2B({\bf{x}}){\mu _1}} {{\widehat {\bf{v}}}_ \bot }} \right)\cdot d{\bf{x}} - \varepsilon\left( {\frac{{u_1^2}}{2} + {\mu _1}B({\bf{x}})} \right)dt, \label{a12_1}\\
\varepsilon ^\sigma {\gamma _{\sigma } } &=&  -\varepsilon^\sigma \phi \left( {{\bf{x}},t} \right)dt.
\end{eqnarray}
\end{subequations}
The $\mathbf{X}$ components in $\gamma_1$ can be decomposed into  the parallel and perpendicular parts as  $\gamma_{1\mathbf{x}\parallel}=\varepsilon u_1\mathbf{b}$ and $\gamma_{1\mathbf{x}\perp}=\varepsilon \sqrt {2B({\bf{x}}){\mu _1}} {{\widehat {\bf{v}}}_ \bot }$.

$\theta$ is a fast variable and the term depending on $\theta$ in Eq.(\ref{a12_1}) is ${\varepsilon }\sqrt {2\mu_1 B({\bf{x}})} {\widehat {\bf{v}}_ \bot }\cdot d{\bf{x}}$ possessing the order $O(\varepsilon)$. $\theta$ can be reduced from the dynamical system up to some order by the coordinate transform.

\subsubsection{The basic orders}\label{order}
There are several basic orders or scales contained by the perturbation. The first one is the length scale of the nondimensionalized Larmor Radius being $\varepsilon$. The second one is the amplitude of the electrostatic potential, whose order is  denoted as $O\left( {{\varepsilon ^\sigma }} \right)$ with the basic parameter $\varepsilon$ as the basis. In magnetized fusion plasmas, due to the fact that the charged particle can nearly migrate freely in the environment with collective interactions, the magnitude of the potential the particles feel must be much smaller than that of its kinetic energy. As Eq.(\ref{g2}) shows, the order of the kinetic energy is $O(\varepsilon)$. Therefore, it's plausible to assume the range for $\sigma$ being $\sigma >1$. In this paper, only
\begin{equation}\label{e56} 
2\le\sigma <3
\end{equation}
is considered. The choice of $2$ is done in Ref.\cite{Hahm1988}. The reason for the choice of the upper bound is given by Eq.(\ref{c17}) in Appendix (\ref{generator})


 The third one is the length scale of the gradient of the electrostatic potential. Define $\mathcal{K}_{\perp}=|\frac{\nabla_\perp\phi}{\phi}|$ and $\mathcal{K}_{\parallel}=|\frac{\nabla_{\parallel}\phi}{\phi}|$. The gyrokinetic model adopts the  scales 
\begin{equation}\label{e57}
O(\varepsilon \mathcal{K}_\perp)=O(1),\;\; O(\varepsilon \mathcal{K}_\parallel)=O(\varepsilon). 
\end{equation}
For any the equilibrium quantity $\mathscr{E}$, the scale 
\begin{equation}\label{e58}
O( \bigg \| \frac{\nabla_\perp\mathscr{E}}{\mathscr{E}}\bigg \| )=O(\bigg \| \frac{\nabla_\parallel \mathscr{E}}{\mathscr{E}}\bigg \| )=O(\bigg \| \frac{\partial_U \mathscr{E}}{\mathscr{E}}\bigg \| )=O(1)
\end{equation}
 is used.

\section{The full-orbit density in SGM not truly accurate at $O(\varepsilon^{\sigma-1})$}\label{error}

\subsection{The transform of the domains of the arguments}\label{error.1}

As explained in Sec.(\ref{introduction}), the gyrokinetic simulations implement $\mu\in (0, \mu_{\max})$ with $\mu$ obeying $\exp(\frac{-\mu B}{T_i})$ to compute the evolution of the gyrocenter distribution for a realistic magnetized plasma.  For the theoretical derivation, $\mu_{\max}$ is usually chosen as $+\infty$.
The transform between $\mu$ and $\bar{\mu}$ is given by Eq.(\ref{ee3.2}) and induces the domain of $\bar{\mu}$
 \begin{equation}
 \bar{\mu}\in(\bar{\mu}_{min}(\bar{\mathbf{X}},\bar{\theta}),+\infty),
\end{equation}
where the upper bound $\bar{\mu}_{max}(\bar{\mathbf{X}},\bar{\theta})$ associated with $\mu=+\infty$  equals $+\infty$.
The transform between $\mu$ and $\mu_1$ induced by $\psi_{gy}$ and $\psi_{gc}$ given by Eqs.(\ref{ee3}) and (\ref{ee2}) is
\begin{equation}\label{e5}
\mu={\mu}_1+ \varepsilon^{\sigma-1}g_2^\mu \left({{\bf{ x}}-\varepsilon \bm{\rho}_0(\mathbf{x},\mu_1,\theta_1),{\mu}_1},{\theta}_1 \right)+O(\varepsilon^{2\sigma-2}).  
\end{equation}
The domain of $\mu_1$ induced by Eq.(\ref{e5}) is denoted as
$${\mu_1}\in ({\mu}_{1min}({\mathbf{x}},\theta_1),+\infty). $$
The domain of $U$, $u_1$ and $\bar{U}$ equals, so does that of $\theta,\bar{\theta},\theta_1$.

\subsection{The order of the error of the density committed by the integral over $\mu_1$ is $O(\varepsilon^{\sigma-1})$}\label{error.3}

In SGM, the density on the spatial space is given by integrating ${f_s}\left( {\bf{z}} \right)$ out of $\mu_1,\theta_1, u_1$
\begin{equation} \label{density}
n_s(\mathbf{x})= \iiint_{{\mu _{1\min}}({\bf{x}},{\theta _1})}^{+\infty}f_s({\bf{z}})B(\mathbf{x})d\mu_1 du_1 d\theta_1 ,
\end{equation}
 where the bounds of the domains of $\theta_1$ and $u_1$ are not explicitly given
and  ${f_s}\left( {\bf{z}} \right)$ is given by Eq.(\ref{qq5}). As Eq.(\ref{e5}) shows, the domain of $\mu_1$ is a function of $(\mathbf{x}, \theta_1)$ for $\mu \in (0,+\infty)$. Because it's a difficult burden to solve the domain of $\mu_1$ at each point $(\mathbf{x},\theta_1)$,
the domain $({\mu _{1\min}}({\bf{x}},{\theta _1}), {\mu _{1\max }}({\bf{x}},{\theta _1}))$ of $\mu_1$ in the standard method is replaced by $(0,+\infty)$.  Meanwhile, the $O(\varepsilon^2)$ term in Eq.(\ref{qq5}) is an uncertain term, the ignorance of which would introduce an error. So there are two errors existing in the density $n_s(\mathbf{x})$ of SGM.  One involves the replacement of the domain of the magnetic moment and the other involves the ignorance of $O(\varepsilon^2)$ term. 

We first estimate the order of the density error due to the ignorance of the uncertain term $O(\varepsilon^2)$, which is temporarily written as $\mathcal{M}(\mathbf{x},\mu_1,u_1,\theta_1)$. The order of the ratio of the error density to total density equals $O(\frac{\int \mathcal{M}(\mathbf{x},\mu_1,u_1,\theta_1) Bd \mu_1 d u_1 d\theta_1}{\int F_0(\mathbf{x},\mu_1,u_1)d\mu_1du_1d\theta_1}) $.  It can be estimated that
\begin{equation}\label{d2}
O\left(\frac{\int \mathcal{M}(\mathbf{x},\mu_1,u_1,\theta_1) Bd \mu_1 d u_1 d\theta_1}{\int F_0(\mathbf{x},\mu_1,u_1)Bd\mu_1du_1d\theta_1}\right)\ge O\left(\frac{\int |\mathcal{M}(\mathbf{x},\mu_1,u_1,\theta_1)| Bd \mu_1 d u_1 d\theta_1}{\int F_0(\mathbf{x},\mu_1,u_1)Bd\mu_1du_1d\theta_1}\right)=O(\varepsilon^2). 
\end{equation}

Now, we estimate the order of the error with respect to the replacement of the domain of the magnetic moment. First of all, the error of this replacement is estimated as
\begin{equation}\nonumber
\mathscr{D}f_s(\mathbf{z})\equiv \left( {\int_{{0}}^{{+\infty}} { - \int_{{\mu _{1\min}}({\bf{x}},{\theta _1})}^{+\infty}  } } \right){f_s}\left( {\bf{z}} \right)B(\mathbf{x})d{\mu _1}.
\end{equation}
where $B(\mathbf{x})$ is the Jacobian due to the transform from the Cartesian $\mathbf{v}$ to  $(\mu_1,\theta_1,u_1)$.
If separating  $f_s(\mathbf{z})$ as an equilibrium one $f_{s0}(\mathbf{z})$ plus a perturbative one $f_{s1}(\mathbf{z})$, then, $$\mathscr{D}f_{s}(\mathbf{z})=\mathscr{D}f_{s0}(\mathbf{z})+\mathscr{D}f_{s1}(\mathbf{z})$$
can be derived. 

\begin{definition}
For a function $f(\varepsilon)$, which depends on a small parameter $\varepsilon$ and can be expanded as $f(\varepsilon)=\sum\limits_{l=m} \frac{\varepsilon^l}{l!}f_l$, with $m\ge 0$ . The leading order term of $f(\varepsilon)$ is denoted as 
$$\mathcal{E}(f(\varepsilon))=\frac{\varepsilon^m}{m!}f_{m}.$$
\end{definition}

The leading order term of  $\mathscr{D}f_{s}(\mathbf{z})$ is $\mathcal{E}(\mathscr{D}f_{s}(\mathbf{z}))$. It's easy to derive that 
$$\mathcal{E}(\mathscr{D}f_{s}(\mathbf{z}))=\mathcal{E}(\mathscr{D}f_{s0}(\mathbf{z})).$$
Due to $f_{s0}(\mathbf{z})=F_{s0}(\mathbf{z})$, the equation
\begin{equation}\label{e42}
f_{s0}(\mathbf{z})=F_{s0\perp}(\mathbf{x},\mu_1)F_{s0\parallel}(\mathbf{x},u_1)
\end{equation}
stands,  so that 
$$\mathscr{D}f_{s0}(\mathbf{z})=F_{s0\parallel}(\mathbf{x},u_1)\mathscr{D}F_{s0\perp}(\mathbf{x},\mu_1)$$
stands.
The error of the density is defined as 
$$n_{serr}(\mathbf{x})\equiv \int \mathscr{D}f_s(\mathbf{z}) du_1 d\theta_1,$$ 
so 
\begin{equation}\label{d1}
\mathcal{E}(n_{serr}(\mathbf{x}))=\mathcal{E}(n_{s0err}(\mathbf{x})),
\end{equation}
where
\begin{equation}\label{e41}
n_{s0err}(\mathbf{x})\equiv \int \mathscr{D}f_{s0}(\mathbf{z}) du_1 d\theta_1=2\pi \bigg (\int F_{s0\parallel}(\mathbf{x},u_1) du_1 \bigg) \mathscr{D}F_{s0\perp}(\mathbf{x},\mu_1) .
\end{equation}
 
Next, the density $n_{s}(\mathbf{x})$ is splited as $n_{s0}(\mathbf{x})+n_{s1}(\mathbf{x})$  with
$$n_{s0}(\mathbf{x})\equiv 2\pi \int F_{s0\parallel}(\mathbf{x},u_1) du_1 \int_{0}^{+\infty} F_{s0\perp}(\mathbf{x},\mu_1)d\mu_1. $$
Then, according to Eq.(\ref{d1}), the leading order term of the ratio of $n_{serr}(\mathbf{x})$ to $n_{s}(\mathbf{x})$ is estimated as
\begin{equation}\label{f5}
\mathcal{E}\bigg(\frac{n_{serr}(\mathbf{x})}{n_{s}(\mathbf{x})}\bigg )=\mathcal{E}\bigg (\frac{n_{s0err}(\mathbf{x})}{n_{s0}(\mathbf{x})}\bigg )=\frac{\mathscr{D}F_{s0\perp}(\mathbf{x},\mu_1)}{\int_{0}^{+\infty} F_{s0\perp}(\mathbf{x},\mu_1)d\mu_1} .
\end{equation}

In the lower bound side, according to Eq.(\ref{e5}), 
$$O(|\mu_{1\min}(\mathbf{x},\theta_1)-0|)=O(|\varepsilon^{\sigma-1}g_2^{\mu}|).$$
Alternatively, the dislocation between $(0,+\infty)$ and $(\mu _{1\min }({\bf{x}},{\theta _1}), +\infty)$ at the lower bound side is of the order $O(\varepsilon^{\sigma-1})$ with respect to a continuous transform given by Eqs.(\ref{ee2}) and (\ref{ee3}). The usually chosen distribution of $\mu$ is  $\exp(-\alpha \mu )$ with $\alpha=\frac{B}{T_i}$ .  Then, 
$$O(\mathscr{D}F_{s0\perp}(\mathbf{x},\mu_1))=O\bigg(\bigg|\int_0^{|\varepsilon^{\sigma-1}g_2^{\mu}|}(1-\alpha \mu)d\mu \bigg|\bigg)
=O(\varepsilon^{\sigma-1})$$ and $$\int_{0}^{+\infty} F_{s0\perp}(\mathbf{x},\mu_1)d\mu_1=\int_{0}^{+\infty}\exp(-\alpha \mu)d\mu\sim O(1)$$ are valid. So it can be estimated that
\begin{equation}\label{e40}
O\bigg(\mathcal{E}\bigg(\frac{n_{serr}(\mathbf{x})}{n_{s}(\mathbf{x})}\bigg )\bigg )=O\bigg(\frac{\mathscr{D}F_{s0\perp}(\mathbf{x},\mu_1)}{\int_{0}^{+\infty} F_{s0\perp}(\mathbf{x},\mu_1)d\mu_1} \bigg)=O(\varepsilon^{\sigma-1}).
\end{equation}
Eventually, by comparing Eq.(\ref{e40}) and (\ref{d2}), due to $O(\varepsilon^{\sigma-1})<O(\varepsilon^2)$,  the error induced by replacing $({\mu}_{1min}({\mathbf{x}},\theta_1),+\infty)$  with $(0,+\infty)$ dominants.
Therefore, the density derived by SGM is not truly accurate at the order $O(\varepsilon^{\sigma-1})$.

\emph{\textbf{Remark}}: The perturbative density contained by QNE is $n(\mathbf{x})-n_0(\mathbf{x})$ with $n(\mathbf{x})\equiv \int f_s(\mathbf{x},\mu_1,u_1,\theta_1) Bd\mu_1 du_1d\theta_1$ and $n_0(\mathbf{x})\equiv \int F_{s0}(\mathbf{x},\mu,U,\theta) B d\mu dUd\theta$. In gyrokinetic simulations, $n_0(\mathbf{x})$ is usually initialized at the beginning. The error of the order $O(\varepsilon^{\sigma-1})$ produced by computing $n(\mathbf{x})$ is inherited by QNE. 

\section{Hybrid coordinate transform and new QNE with exact order $O(\varepsilon^{\sigma-1})$ } \label{new}

%

\subsection{The full-orbit density with the exact order $O(\varepsilon^{\sigma-1})$ }\label{new.2}

Given the coordinate transform Eqs.(\ref{ee2}) and (\ref{ee3}), the exact full-orbit distribution is given by Eq.(\ref{e36}). To prevent the error pointed out by Subsec.(\ref{error.3}), the expression ${\mu}_1- \varepsilon^{\sigma-1}g_2^\mu \left({{\bf{ x}}-\varepsilon \bm{\rho}_0(\mathbf{x},\mu_1,\theta_1),{\mu}_1},{\theta}_1 \right)$ is inversely replaced by $\mu$ with respect to Eq.(\ref{e5}) and $\mu_1$ can be solved as a function of $(\mathbf{x},\mu,\theta_1)$. Therefore, $f_s(\mathbf{z})$ can be rewritten as a function of the hybrid coordinates $(\mathbf{x},\mu,u_1,\theta_1)$ and is denoted as $f_s^*(\mathbf{x},\mu,u_1,\theta_1)$ with
\begin{equation} \label{e21}
f_s^*(\mathbf{x},\mu,u_1,\theta_1)\equiv{F}_{s}({\mathbf{x}}-\varepsilon \bm{\rho}_0(\mathbf{x},\mu_1(\mathbf{x},\mu,\theta_1),\theta_1),\mu,u_1)+O(\varepsilon^2),
\end{equation}
where the uncertain term $O(\varepsilon^2)$ is inherited from the $O(\varepsilon^2)$ term in Eq.(\ref{ee2.1}). 

On the coordinate frame of $\bar{\mathbf{z}}\equiv (\mathbf{x},\mathbf{v})$, the infinitesimal volume element of the velocity space is $d^3\mathbf{v}$. By transforming  $\bar{\mathbf{z}}$ to the coordinate frame of ${\mathbf{z}}\equiv (\mathbf{x},\mu_1,u_1,\theta_1)$, the normalized infinitesimal volume element for the subspace parameterized by $(\mu_1,u_1,\theta_1)$ is 
$$B(\mathbf{x})d\mu_1du_1d\theta_1.$$
On the frame of ${\mathbf{z}}\equiv (\mathbf{x},\mu_1,u_1,\theta_1)$, the spatial density is given by Eq.(\ref{density}).
On the hybrid coordinate frame $(\mathbf{x},\mu,u_1,\theta_1)$,  the normalized infinitesimal volume element changes to be
\begin{equation}\nonumber
B(\mathbf{x})\frac{\partial \mu_1}{\partial \mu}d\mu du_1d\theta_1,
\end{equation}
where the mutual independence of $\mathbf{x}, \mu_1,u_1,\theta_1$ is used and $\frac{\partial \mu_1}{\partial \mu}$ is the Jacobian.
So the full-orbit spatial density becomes
\begin{align} \label{e37}
n_s(\mathbf{x})=\iiint_{\mu_{\min}=0}^{{\mu_{\max}}=+\infty}f_s^*(\mathbf{x},\mu,u_1,\theta_1)B(\mathbf{x})\frac{\partial \mu_1}{\partial \mu}d\mu du_1d\theta_1.
\end{align}
The approximation of Eq.(\ref{e37}) can be obtained through the approximation of $\mu_1(\mathbf{x},\mu,\theta_1)$.
\begin{prop}\label{prop8}
Given the equation of $\mu$ in Eq.(\ref{e5}),  $\mu_1$ as a function of $(\mu,\mathbf{x},\theta_1)$ can be solved  with the exact order  being $O(\varepsilon^{\sigma-1})$ 
\begin{equation}\label{e15}
\mu_1(\mu,\mathbf{x},\theta_1)= {\mu^{*}}+O(\varepsilon^{2\sigma-2}).
\end{equation}
where
\begin{equation}\label{e28}
{\mu^{*}}\equiv {\mu}-\varepsilon^{\sigma-1}g_2^\mu \left({{\bf{ x}}-\varepsilon \bm{\rho}_0(\mathbf{x},\mu,\theta_1),{\mu}},{\theta}_1 \right).
\end{equation}
\end{prop}

\begin{proof}
Rewrite Eq.(\ref{e5}) 
\begin{equation}\label{e13}
\mu_1={\mu}-\varepsilon^{\sigma-1}g_2^\mu \left({{\bf{ x}}-\varepsilon \bm{\rho}_0(\mathbf{x},\mu_1,\theta_1),{\mu}_1},{\theta}_1 \right)+O(\varepsilon^{2\sigma-2}).
\end{equation}
Iterating $\mu_1$ one time in Eq.(\ref{e13}) and expanding $g_2^\mu$  in Eq.(\ref{e13}) by the order parameter $\varepsilon^{\sigma}$, noticing $O(\varepsilon \mathcal{K}_\perp)=1$, $g_2^\mu$ can be written as
\begin{equation}
g_2^\mu \left({{\bf{ x}}-\varepsilon \bm{\rho}_0(\mathbf{x},\mu_1,\theta_1),{\mu}_1},{\theta}_1 \right)=g_2^\mu \left({{\bf{ x}}-\varepsilon \bm{\rho}_0(\mathbf{x},\mu,\theta_1),{\mu}},{\theta}_1 \right) + O(\varepsilon^{\sigma-1}), \nonumber
\end{equation}
whose exact order is $O(1)$ and uncertain order is $O(\varepsilon^{\sigma-1})$. By substituting this equation into Eq.(\ref{e13}), Eq.(\ref{e15}) is derived.
\end{proof}

Given Proposition.(\ref{prop8}), $\frac{\partial \mu_1(\mu,\mathbf{x},\theta_1)}{\partial \mu}$ can be written as
\begin{equation}\label{e39} 
\frac{\partial \mu_1(\mu,\mathbf{x},\theta_1)}{\partial \mu}=1-\varepsilon^{\sigma-1}\frac{\partial g_2^\mu \left({{\bf{ x}}-\varepsilon \bm{\rho}_0(\mathbf{x},\mu,\theta_1),{\mu}},{\theta}_1 \right)}{\partial \mu}+O(\varepsilon^{2\sigma-2}), 
\end{equation}
with the exact order being $O(\varepsilon^{\sigma-1})$.

\begin{prop}\label{prop9}
Given Proposition.(\ref{prop8}), $\bm{\rho}_0(\mathbf{x},\mu_1,\theta_1)$ can be solved with the exact order being $O(\varepsilon^{\sigma-1})$:
\begin{equation} \label{e20}
\bm{\rho}_0(\mathbf{x},\mu_1,\theta_1)=\bm{\rho}_0(\mathbf{x},\mu^*,\theta_1)+O(\varepsilon^{2\sigma-2}).
\end{equation}
\end{prop}
\begin{proof}
By substituting Eq.(\ref{e15}) into $\bm{\rho}_0(\mathbf{x},\mu_1,\theta_1)$ and using $O(\varepsilon \mathcal{K}_\perp)=1$, Eq.(\ref{e20}) can be derived.
\end{proof}


\begin{prop}\label{prop4}
Given Proposition.(\ref{prop9}) and $3>\sigma\ge 2$, the distribution $f_s^*(\mathbf{x},\mu,u_1,\theta_1)$ in Eq.(\ref{e21}) can be solved as:
\begin{equation}\label{e24}
f_s^*(\mathbf{x},\mu,u_1,\theta_1)={F}_{s}({\mathbf{x}}-\varepsilon \bm{\rho}_0(\mathbf{x},\mu^*,\theta_1),\mu,u_1)+ O(\varepsilon^{2}),
\end{equation}
with the exact order being $O(\varepsilon^{\sigma-1})$.
\end{prop}
\begin{proof}
It's first to prove the following two statements:
\begin{equation}\label{e22}
f_{s0}^*(\mathbf{x},\mu,u_1,\theta_1)={F}_{s0}({\mathbf{x}}-\varepsilon \bm{\rho}_0(\mathbf{x},\mu^*,\theta_1),\mu,u_1)+ O(\varepsilon^{2}),
\end{equation}
the exact order of which is $O(\varepsilon^{\sigma})$, and
\begin{equation}\label{e23}
f_{s1}^*(\mathbf{x},\mu,u_1,\theta_1)={F}_{s1}({\mathbf{x}}-\varepsilon \bm{\rho}_0(\mathbf{x},\mu^*,\theta_1),\mu,u_1)+ O(\varepsilon^{2}),
\end{equation}
the exact order of which is $O(\varepsilon^{2\sigma-1})$.

By substituting $\bm{\rho}_0(\mathbf{x},\mu_1,\theta_1)$ in Eq.(\ref{e20}) to $f_{s0}^*(\mathbf{x},\mu,u_1,\theta_1)$, Eq.(\ref{e22}) is proved.
For Eq.(\ref{e23}),   by substituting $\bm{\rho}_0(\mathbf{x},\mu_1,\theta_1)$  in Eq.(\ref{e20})  to $f_{s1}^*(\mathbf{x},\mu,u_1,\theta_1)$
and considering $O(\rho \mathcal{K}_\perp)=1$,  Eq.(\ref{e23}) is proved.
By combining Eqs.(\ref{e22}) and (\ref{e23}), Eq.(\ref{e24}) is obtained.
\end{proof}

To solve $\phi$ through QNE, the approximation of $n_s(\mathbf{x})$  is required.
\begin{theorem}\label{theo.1}
$n_s(\mathbf{x})$ can be approximated as
\begin{equation} \label{e38} 
\begin{split}
n_s(\mathbf{x})=\iiint_{\mu_{\min}=0}^{\mu_{\max}=+\infty}
\bigg [ \begin{array}{l}
{F}_{s}({\mathbf{x}}-\varepsilon \bm{\rho}_0(\mathbf{x},\mu^*,\theta_1),\mu,u_1)  \\
-\mathscr{F}({F}_{s}({\mathbf{x}}-\varepsilon \bm{\rho}_0(\mathbf{x},\mu,\theta_1),\mu,u_1),\mu,u_1, \theta_1) 
\end{array} \bigg ] B(\mathbf{x})d\mu du_1 d\theta_1   +O(\varepsilon^{2})
\end{split}
\end{equation}
with the exact order being $O(\varepsilon^{\sigma-1})$, where
\begin{equation}\label{e45}
\begin{split}
\mathscr{F}({F}_{s}({\mathbf{x}}-\varepsilon \bm{\rho}_0(\mathbf{x},\mu,\theta_1),\mu,u_1),\mu,u_1, \theta_1) &\equiv \varepsilon^{\sigma-1}{F}_{s}({\mathbf{x}}-\varepsilon \bm{\rho}_0(\mathbf{x},\mu,\theta_1),\mu,u_1)  \\
& \times\frac{\partial g_2^\mu \left({{\bf{ x}}-\varepsilon \bm{\rho}_0(\mathbf{x},\mu,\theta_1),{\mu}},{\theta}_1 \right)}{\partial \mu} 
\end{split}
\end{equation}

\end{theorem}

\begin{proof}
Based on Eq.(\ref{e39}), $f_s^*(\mathbf{x},\mu,u_1,\theta_1)\frac{\partial \mu_1}{\partial \mu}$ can be approximated as the sum
$$f_s^*(\mathbf{x},\mu,u_1,\theta_1)+\mathscr{F}({f}_{s}^*(\mathbf{x},\mu,u_1,\theta_1),\mu_1,u_1, \theta_1)+O(\varepsilon^{2\sigma-2})$$
with the exact order being $O(\varepsilon^{\sigma-1})$. 
According to Proposition.(\ref{prop4}), $f_s^*(\mathbf{x},\mu,u_1,\theta_1)$ can be approximated as ${F}_{s}({\mathbf{x}}-\varepsilon \bm{\rho}_0(\mathbf{x},\mu^*,\theta_1),\mu,u_1)$ with the exact order being $O(\varepsilon^{2})$. 
The second term can be approximated as $\mathscr{F}({F}_{s}({\mathbf{x}}-\varepsilon \bm{\rho}_0(\mathbf{x},\mu,\theta_1),\mu,u_1),\mu,u_1,\theta_1) +O(\varepsilon^{\sigma+1})$ exactly right up to $O(\varepsilon^{\sigma})$.  Then, $f_s^*(\mathbf{x},\mu,u_1,\theta_1)\frac{\partial \mu_1}{\partial \mu}$ can be rewritten as
$${F}_{s}({\mathbf{x}}-\varepsilon \bm{\rho}_0(\mathbf{x},\mu^*,\theta_1),\mu,u_1)+\mathscr{F}({F}_{s}({\mathbf{x}}-\varepsilon \bm{\rho}_0(\mathbf{x},\mu,\theta_1),\mu,u_1),\mu,u_1,\theta_1) +O(\varepsilon^{2}), $$
with the exact order being $\varepsilon^{\sigma-1}$ and uncertain order being $\varepsilon^{2}$.
As a consequence, theorem.(\ref{theo.1}) can be proved in the same way to prove the inequality (\ref{d2}). 
\end{proof}

The term of ${F}_{s}({\mathbf{x}}-\varepsilon \bm{\rho}_0(\mathbf{x},\mu^*,\theta_1),\mu,u_1)$ in Eq.(\ref{e38}) depends on $\phi$ through $\sqrt{\mu+\varepsilon^{\sigma-1}g_2^{\mu}}$, which makes the solving of $\phi$ not convenient through QNE and needs to be simplified to be linearly proportional to $\phi$. 

\begin{prop}\label{prop2.6}
If $O(\mu)< O(\varepsilon^{\sigma-1})$ holds for the number of $\mu$, specifically,  $\mu>|\varepsilon^{\sigma-1}g_2^\mu \left({{\bf{ x}}-\varepsilon \bm{\rho}_{0}(\mathbf{x},\mu,\theta_1),{\mu}},{\theta}_1 \right)|$ holds, the expansion of $\bm{\rho}_0(\mathbf{x},\mu^*,\theta_1)$ with the exact order being $O(\frac{\varepsilon^{\sigma-1}}{\mu})$ is
$$\bm{\rho}_0(\mathbf{x},\mu^*,\theta_1) = \bm{\rho}^*(\mathbf{x},\mu,\theta_1) +O(\frac{\varepsilon^{2\sigma-2}}{\mu^2})$$ with
\begin{equation} \label{e32}
\bm{\rho}^*(\mathbf{x},\mu,\theta_1) = \big (1-\frac{\varepsilon^{\sigma-1}g_2^\mu \left({{\bf{ x}}-\varepsilon {\bm{\rho} _{0}(\mathbf{x},\mu,\theta_1)},{\mu}},{\theta}_1 \right) }{2\mu} \big ) {\bm{\rho} _{0}(\mathbf{x},\mu,\theta_1)} .
\end{equation}
\end{prop}
\begin{proof}
By expanding  $\bm{\rho}_0(\mathbf{x},\mu^*,\theta_1)$ over the parameter $\varepsilon^{\sigma-1}$, Eq.(\ref{e32}) is derived. 
\end{proof}

\begin{prop}\label{prop3.6}
The integral $\int_{0}^{+\infty}{F}_{s}({\mathbf{x}}-\varepsilon \bm{\rho}_0(\mathbf{x},\mu^*,\theta_1),\mu,u_1)d\mu$ 
can be written as 
\begin{equation}\label{e51}
\int_{0}^{+\infty}{F}_{s}({\mathbf{x}}-\varepsilon \bm{\rho}_0(\mathbf{x},\mu,\theta_1),\mu,u_1)d\mu+O(\varepsilon^{\sigma} \ln\varepsilon^{\sigma-1}),
\end{equation}
with the exact order being $O(\varepsilon^{\sigma-1})$.
\end{prop}

\begin{proof}
$\int_{0}^{+\infty}{F}_{s}({\mathbf{x}}-\varepsilon \bm{\rho}_0(\mathbf{x},\mu^*,\theta_1),\mu,u_1)d\mu$ is splitted as the sum of two parts
$$\underbrace{\int_{0}^{\mu_\sigma}{F}_{s}({\mathbf{x}}-\varepsilon \bm{\rho}(\mathbf{x},\mu^*,\theta_1),\mu,u_1)d\mu}_{\textbf{(1)}}+\underbrace{\int_{\mu_{\sigma}}^{+\infty}{F}_{s}({\mathbf{x}}-\varepsilon \bm{\rho}(\mathbf{x},\mu^*,\theta_1),\mu,u_1)d\mu}_{\text{(2)}}.$$ Here,
$\mu_\sigma \equiv |\varepsilon^{\sigma-1}g_2^\mu \left({{\bf{ x}}-\varepsilon \bm{\rho}_{0}(\mathbf{x},\mu,\theta_1),{\mu}},{\theta}_1 \right)|$.

Term ``$(1)$" can be rewritten as
$$\int_{0}^{\mu_\sigma}{F}_{s}({\mathbf{x}}-\varepsilon \bm{\rho}_0(\mathbf{x},\mu,\theta_1),\mu,u_1)d\mu+O(\varepsilon^\sigma)$$  which is exactly correct at $O(\varepsilon^{\sigma-1})$. The order of the error is  determined by $O(\mu_\sigma|\varepsilon \bm{\rho}_0(\mathbf{x},\mu,\theta_1)\cdot \nabla {F}_{s0}|)\sim O(\varepsilon^\sigma)$.

In the domain $(\mu_\sigma,+\infty)$, according to Proposition.(\ref{prop2.6}), ${F}_{s}({\mathbf{x}}-\varepsilon \bm{\rho}^*(\mathbf{x},\mu,\theta_1),\mu,u_1)$ can be expanded with the order parameter $\varepsilon^{\sigma}$, which is independent of $\varepsilon$. The truncation of the expansion at the linear term is
\begin{equation}\label{e43}
\begin{split}
{F}_{s}({\mathbf{x}}-\varepsilon \bm{\rho}^*(\mathbf{x},\mu,\theta_1),\mu,u_1) &= \bigg [1+
\frac{\varepsilon^{\sigma}g_2^\mu \left({{\bf{ x}}-\varepsilon {\bm{\rho} _{0}(\mathbf{x},\mu,\theta_1)},{\mu}},{\theta}_1 \right) }{2\mu}  {\bm{\rho} _{0}(\mathbf{x},\mu,\theta_1)} \cdot \nabla\bigg]   \\
&  \times {F}_{s}({\mathbf{x}}-\varepsilon \bm{\rho}_0(\mathbf{x},\mu,\theta_1),\mu,u_1)+O(\frac{\varepsilon^{\sigma+1}}{\mu})
\end{split}
\end{equation}
with the exact order being $O(\frac{\varepsilon^\sigma}{\mu})$.  
Define the functional 
\begin{equation}\nonumber 
\begin{split}
 & \mathscr{A}({F}_{s}({\mathbf{x}}-\varepsilon \bm{\rho}_0(\mathbf{x},\mu,\theta_1),\mu,u_1),\beta_1,\beta_2) \\
& \equiv \int_{\mu_{\min}=\beta_1}^{\mu_{\max}=\beta_2}
\bigg[\begin{array}{l} 
\frac{\varepsilon^{\sigma}g_2^\mu \left({{\bf{ x}}-\varepsilon {\bm{\rho} _{0}(\mathbf{x},\mu,\theta_1)},{\mu}},{\theta}_1 \right) }{2\mu}  {\bm{\rho} _{0}(\mathbf{x},\mu,\theta_1)}  \\
\cdot \nabla {F}_{s}({\mathbf{x}}-\varepsilon \bm{\rho}_0(\mathbf{x},\mu,\theta_1),\mu,u_1)
\end{array} \bigg] d\mu . 
\end{split}
\end{equation}
Since $O(|F_{s1}|)=O(\varepsilon^{\sigma-1})$ and $O(||\frac{\nabla_\perp F_{s0}}{F_{s0}}||)=O(1)$, it's obtained that
$$O(\mathcal{E}(\mathscr{A}({F}_{s}({\mathbf{x}}-\varepsilon \bm{\rho}_0(\mathbf{x},\mu,\theta_1),\mu,u_1),\mu_\sigma,+\infty)))=O(\mathcal{E}(\mathscr{A}({F}_{s0}({\mathbf{x}},\mu,u_1),\mu_\sigma,+\infty))).$$
Concerning the equilibrium perpendicular distribution $\exp(-\alpha\mu)$, 
$O(\mathcal{E}(\mathscr{A}({F}_{s0}({\mathbf{x}},\mu,u_1), \\
\mu_\sigma,+\infty)))$ can be estimated as
$$O(\varepsilon^{\sigma}\int_{\mu_{\sigma}}^{+\infty}\frac{\exp(-\alpha\mu)}{\mu})=O(\varepsilon^{\sigma} \int_{\mu_{\sigma}}^{1}\frac{1}{\mu}d\mu)=O(\varepsilon^{\sigma}\ln\varepsilon^{\sigma-1})>O(\varepsilon^{\sigma-1}).$$
Therefore, the ignorance of the second term of Eq.(\ref{e43}) only introduces an error of the order $O(\varepsilon^{\sigma}\ln\varepsilon^{\sigma-1})$. 

Combing the rest terms of term ``$(1)$'' and  term ``$(2)$'', Eq.(\ref{e51}) is derived. 

\end{proof}

At last, the following corollary is achieved: 

\begin{corollary}\label{coro}
$n_s(\mathbf{x})$ in Eq.(\ref{e38}) can be approximated as 
\begin{equation}\label{e44}
\begin{split}
n_s(\mathbf{x}) &=\iiint_{\mu_{\min}=0}^{\mu_{\max}=+\infty}\bigg [
\begin{array}{l}
{F}_{s}({\mathbf{x}}-\varepsilon \bm{\rho}_0(\mathbf{x},\mu,\theta_1),\mu,u_1) \\
+\mathscr{F}({F}_{s0}({\mathbf{x}},\mu,u_1),\mu,u_1,\theta_1)
\end{array}
 \bigg ] B(\mathbf{x})d\mu du_1 d\theta_1 + O(\varepsilon^{\mathscr{I}})
\end{split}
\end{equation}
with the exact order being $O(\varepsilon^{\sigma-1})$, where the uncertain term possesses the order
$$O(\varepsilon^{\mathscr{I}})=\min\{O(\varepsilon^{\sigma}\ln \varepsilon^{\sigma-1}), O(\varepsilon^2)\},$$
and 
\begin{equation}\nonumber
\begin{split}
\mathscr{F}({F}_{s0}({\mathbf{x}},\mu,u_1),\mu,u_1,\theta_1) = \frac{\varepsilon^{\sigma-1}{F}_{s0}({\mathbf{x}},\mu,u_1)}{B(\mathbf{x})}  
\frac{ \partial \Phi \left({{\bf{ x}}-\varepsilon \bm{\rho}_0(\mathbf{x},\mu,\theta_1),{\mu}}\right)}{\partial \mu} .
\end{split}
\end{equation}
\end{corollary}

\begin{proof}
The reduction of $\mathscr{F}({F}_{s}({\mathbf{x}}-\varepsilon \bm{\rho}_0(\mathbf{x},\mu,\theta_1),\mu,u_1),\mu,u_1,\theta_1)$ in Eq.(\ref{e38}) to \\ 
$\mathscr{F}({F}_{s0}({\mathbf{x}}-\varepsilon \bm{\rho}_0(\mathbf{x},\mu,\theta_1),\mu,u_1),\mu,u_1,\theta_1)$ only introduces an error of the order $O(\varepsilon^{2\sigma-2})$ due to $O(F_{s1})=O(\varepsilon^{\sigma-1})$. The further approximation of $\mathscr{F}({F}_{s0}({\mathbf{x}}-\varepsilon \bm{\rho}_0(\mathbf{x},\mu,\theta_1),\mu,u_1),\mu,u_1)$ to $\mathscr{F}({F}_{s0}({\mathbf{x}},\mu,u_1),\mu,u_1)$ introduces an error of the order $O(\varepsilon^{\sigma})$. 

The approximation of the integrand of Eq.(\ref{e37}) to that of Eq.(\ref{e44}) introduces two error terms, one of which would become the uncertain term contained by $f_s^*(\mathbf{x},\mu,u_1,\theta_1)$ in Eq.(\ref{e21}). The second one is of the order $O(\varepsilon^{\sigma}\ln \varepsilon^{\sigma-1})$  proved by Proposition.(\ref{prop3.6}).  So, the uncertain term is $\min\{O(\varepsilon^{\sigma}\ln \varepsilon^{\sigma-1}), O(\varepsilon^2)\}$.  Then, Eq.(\ref{e44}) can be proved in the same way to prove the inequality (\ref{d2}).
\end{proof}
As a consequence of Corollary.(\ref{coro}), compared with the density in SGM which is not truly accurate at order $O(\varepsilon^{\sigma-1})$, the density in Eq.(\ref{e44}) is exactly correct at  $O(\varepsilon^{\sigma-1})$.


\subsection{The non-normalized QNE of the new model} 

First, the units of all the quantities are recovered.  The Larmor radius with the units recovered is denoted as
\begin{equation}\label{extr}
{\bar{\bm{\rho}} _0} (\mathbf{x},\mu,\theta_1)= \frac{1}{q_s}\sqrt {\frac{{2m_s\mu }}{{B\left( {{{\bf{x}}}} \right)}}} \left( { - {{\bf{e}}_1}\cos \theta_1  + {{\bf{e}}_2}\sin \theta_1 } \right).
\end{equation}
The plasma concerned here only contains  electrons and one species ion being protons. For the equilibrium distribution given by Subsec.(\ref{distri}), 
based on the density in Eq.(\ref{e44}), QNE with unites recovered is
\begin{subequations}\label{e52}
\begin{eqnarray}
&&-n_{i1}-\tilde{\Phi} '+\frac{en_{0}}{T_e}\phi=0,  \label{e52.1} \\
&& n_{i1}=\iiint_{\mu_{\min}=0}^{\mu_{\max}} F_{i1}(\mathbf{x}-\bar{\bm{\rho}}_0(\mathbf{x},\mu,\theta_1),\mu,u_1)\frac{B}{m_i}d\mu du_1 d\theta_1, \label{e52.2} \\
&& \tilde{\Phi} '=\frac{em_i n_{0}}{2\pi T_i B}\iint_{\mu_{\min}=0}^{\mu_{\max}} \exp(\frac{-\mu B}{T_i})\frac{\partial \Phi(\mathbf{x}-\bar{\bm{\rho}}_0(\mathbf{x},\mu,\theta_1),\mu)}{\partial \mu} \frac{B}{m_i} d\mu d\theta_1. \label{e52.3}
\end{eqnarray}
\end{subequations}
Here, since $\mu$ is a conserved quantity and the equilibrium distribution is proportional to $\exp(\frac{-\mu B}{T_i})$, the upper bound of the domain for $\mu$
is not necessary to be $+\infty$ for the realistic application. So $\mu_{\max}$ is used to replace $+\infty$ in the up equations.

\section{The gyrokinetic models}\label{model}


In this simulation, the $\theta$-pinch magnetic field configuration is used with constant amplitude of the magnetic field in the simulated region. So the cylindrical coordinates frame will be used. The numerical solutions are computed using normalized equations. The quantities $t,v,B,l,\mu,T,\phi$ are normalized by $t_0\equiv \frac{m}{B_0 q_i}$, $v_0\equiv \sqrt{\frac{T_{e0}}{m_i}}$, $B_0$, $l_0\equiv\frac{mv_0}{eB_0}$, $\mu_0\equiv \frac{T_{e0}}{B_0}$, $T_{e0}$ and $\phi_0\equiv \frac{T_{e0}}{q_i}$, respectively, where  $T_{e0}\equiv T_e(r_p)$ and $r_p \in [r_{min},r_{max}]$ is the radial coordinate of the peak of the initial distribution function.\\

\setlength\parindent{0em}  \textbf{QNE of the new model} :  \\
\setlength\parindent{1.5em}
The normalized version of Eq.(\ref{e52.1}) is
\begin{equation}\label{e35}
-\frac{\tilde{{\Phi}}(\mathbf{x})}{T_i}+\frac{\phi(\mathbf{x})}{T_e}=\frac{ n_{i1}}{n_0}
\end{equation}
with 
\begin{subequations}\label{e53}
\begin{eqnarray}
&&\tilde{{\Phi}}(\mathbf{x})=\frac{1}{2\pi}\iint_{0}^{\mu_{\max}}\exp(\frac{-\mu B(\mathbf{x})}{T_i(\mathbf{x})}) \frac{\partial \Phi(\mathbf{x}-\bar{\bm{\rho}}'_0(\mathbf{x},\mu,\theta_1),\mu)}{\partial \mu} B(\mathbf{x})d\mu d\theta_1, \label{e53a}\\
&&n_{i1}(\mathbf{x})= \iiint_{0}^{\mu_{\max}} F_{i1}(\mathbf{x}-\bar{\bm{\rho}}'_0(\mathbf{x},\mu,\theta_1),\mu,u_1)B(\mathbf{x})d\mu du_1 d\theta_1,  \label{e53b} \\
&&\bar{\bm{\rho}}'_0(\mathbf{x},\mu,\theta_1)=\sqrt{\frac{2\mu}{B(\mathbf{x})}}\left( { - {{\bf{e}}_1}\cos \theta_1  + {{\bf{e}}_2}\sin \theta_1 } \right). \label{e53c}
\end{eqnarray}
\end{subequations}
\\
\setlength\parindent{0em}  \textbf{QNE of the standard model} :  \\
\setlength\parindent{1.5em}
The normalized QNE of the standard model can be written as
\begin{equation}\label{e59}
\frac{{\phi ({\bf{x}})}}{{{T_i}}} - \frac{{B{\tilde\phi} ({\bf{x}})}}{{T_i^2}} + \frac{\phi ({\bf{x}}) }{{{T_e}}}= \frac{{{n_{i1}}}}{{{n_{i0}}}} ,
\end{equation}
where $n_{i1}$ is given by Eq.(\ref{e53b}) and  ${\tilde\phi} ({\bf{x}})$ is 
\begin{equation}
{\tilde\phi} ({\bf{x}})=\frac{1}{2\pi }\iint_{0}^{\mu_{\max}}\Phi\big(\mathbf{x}-\bar{\bm{\rho}}'_0\big(\mathbf{x},\mu_1,\theta_1\big),\mu_1\big)\exp\bigg(-\frac{\mu_1 B}{T_i}\bigg)B({\bf{x}})d\mu_1d\theta_1
\end{equation}
\\
\setlength\parindent{0em}  \textbf{The equations of motion and Vlasov equation} :  \\
\setlength\parindent{1.5em}
The normalized orbit equations of the gyrocenter coordinates are
\begin{subequations}\label{b34}
\begin{eqnarray}
\dot {\bf{X}} ({\bf{X}},\mu,U) &=& \frac{{U{{\bf{B}}^*} - {\bf{b}} \times \nabla (\mu B + \Phi ({\bf{X}},\mu ))}}{{\mathbf{b}\cdot \mathbf{B} ^*}}, \label{b34.1} \\
\dot U ({\bf{X}},\mu,U)  &=& \frac{{{{\bf{B}}^*}\cdot\nabla \left( {\mu B + \Phi ({\bf{X}},\mu )} \right)}}{{\mathbf{b}\cdot \mathbf{B} ^*}}, \label{b34.2}  \\
\dot{\mu}&=&0 \label{b34.3}. 
\end{eqnarray}
\end{subequations}
where $\mathbf{B}^* \equiv\mathbf{B}+U\nabla\times \mathbf{b}=1 \mathbf{e}_\parallel$ due to the choice of $\mathbf{B}=1 \mathbf{e}_\parallel$. With $\mathbf{B}=1 \mathbf{e}_\parallel$, it's easy to check the incompressible property of the orbit equation 
\begin{equation}\label{b35}
\nabla \cdot\mathop {\bf{X}}\limits^.  + {\partial _U}\dot U = 0.
\end{equation}
Then, the Vlasov equation can be rewritten in a flux form
\begin{equation}\label{b36}
\frac{{\partial F({\bf{X}},\mu,U)}}{{\partial t}} + \frac{d}{{d{\bf{X}}}}\cdot\left( {\mathop {\bf{X}}\limits^. F({\bf{X}},\mu,U)} \right) + \frac{d}{{dU}}\left( {\dot UF({\bf{X}},\mu,U)} \right) = 0.
\end{equation}
In the numerical simulation, $\mathop {\bf{X}}\limits^. $ and $\dot U$ will be formulated in the cylindrical coordinate frame.

\section{The numerical simulation}\label{numerical}

Since the spatial domain of the full-orbit coordinate frame and gyrocenter coordinate frame is identical, we will use the symbol $\mathbf{x}$ uniformly to denote the spatial domain. 

\subsection{The formulas in cylindrical coordinates}

The theta-pinch magnetic field configuration with the constant amplitude is implemented in this simulation. 
The equations of motion given by Eqs.(\ref{b34.1},\ref{b34.2},) are rewritten in cylindrical coordinate frame as
\begin{subequations}
\begin{eqnarray}
\dot{r} &=& \frac{1}{r}\partial_\Theta \Phi, \\
\dot{\Theta}&=&-\frac{1}{r}\partial_r \Phi, \\
\dot{x_\parallel}&=& U, \\
\dot{U}&=&\partial_{x_\parallel} \Phi
 \end{eqnarray}
\end{subequations}
The Vlasov equation in cylindrical coordinates is
\begin{equation}\label{e70}
\big[{\partial_t}+(\frac{1}{r}\partial_\Theta \Phi {\partial_r}-\frac{1}{r}\partial_r \Phi {\partial_\Theta})+ U \partial_ {x_\parallel}+ \partial_{x_\parallel} \Phi {\partial_U}\big]F=0.
\end{equation}

\subsection{The algorithms used in this simulation}

\subsubsection{The algorithm with respect to $\mu$}

The domain $(0,\mu_{\max})$  is divided into $N_\mu$ segments with unequal length by the following scheme. We choose a weight function $\exp(-\frac{\mu B}{T_i(r_0)})$ with $r_0\equiv\frac{r_{\min}+r_{\max}}{2}$ and require that the neighbour  points satisfy the equation $\mathscr{G}(\mu_{j-1},\mu_{j})=0$ for $j\ge 2$ with the function $\mathscr{G}(\mu_{j-1},\mu_{j})$ defined as
\begin{equation}\label{e63}
\mathscr{G}(\mu_{j-1},\mu_{j})\equiv \int_{\mu_{j-1}}^{\mu_{j}}e^{\frac{-\mu B}{T_i(r_0)}}d\mu-\frac{\int_{0}^{\mu_{\max}}e^{\frac{-\mu B}{T_i(r_0)}}d\mu}{N_\mu}.
\end{equation}
The first point $\mu_{1}$ satisfies $\mathscr{G}(0, \mu_{1})=0$.  The step length for $j$ with $N_\mu-1\ge  j\ge 2$ is defined as
\begin{equation}\label{e64}
\delta\mu_{j}=\frac{\mu_{j+1}-\mu_{j-1}}{2},
\end{equation}
while $\delta\mu_{1}=\frac{\mu_{2}}{2}$ and $\delta \mu_{N_\mu}=\frac{\mu_{\max}-\mu_{N_\mu-1}}{2}$.

In the discrete version of $\mu$, the distribution of ions associated with each $\mu_{j}$ with $j\in\{1,\cdots,N_\mu\}$ is denoted as $F(\mathbf{x},\mu_{j},U)$.
Due to  the identity $\frac{d\mu_j}{dt}=0$, $F(\mathbf{x},\mu_{j},U)$ satisfies the Vlasov equation 
\begin{equation}\label{e68}
\big[{\partial_t}+(\frac{1}{r}\partial_\Theta \Phi_j {\partial_r}-\frac{1}{r}\partial_r \Phi_j {\partial_\Theta})+ U \partial_ {x_\parallel}+ \partial_{x_\parallel} \Phi_j {\partial_U}\big]F_i(\mathbf{x},\mu_j,U)=0.
\end{equation}
$F_{i}(\mathbf{x},\mu_{j},U)$ can be rewritten as the sum 
$$F_{i}(\mathbf{x},\mu_{j},U)=F_{0i}(\mathbf{x},\mu_{j},U)+F_{1i}(\mathbf{x},\mu_{j},U), $$
with 
\begin{equation}\label{e67}
F_{0i}(\mathbf{x},\mu_{j},U)=F_{0i\parallel}(\mathbf{x},U)F_{0i\perp}(\mathbf{x},\mu_j)
\end{equation}
and  $F_{0j\perp}(\mathbf{x},\mu_j)=\frac{1}{2\pi T_i}\exp(-\frac{\mu_j B(\mathbf{x})}{T_i(\mathbf{x})})$. In the numerical simulation, $F_{0i}(\mathbf{x},\mu_{j},U)$ doesn't evolve. At each time step, $F_{i}(\mathbf{x},\mu_{j},U)$ is obtained by solving Eq.(\ref{e68}) and $F_{1i}(\mathbf{x},\mu_{j},U)$ is derived by using $F_{i}(\mathbf{x},\mu_{j},U)$ minus $F_{0i}(\mathbf{x},\mu_{j},U)$.

The full-orbit distribution associated with each $j$ is denoted as $f_i(\mathbf{x},\mu_j,u_1,\theta)$. For the new model, its contribution to the density on the full-orbit coordinate frame is contained by $\tilde{\Phi}(\mathbf{x},\mu_j)$ and $n_{i1}(\mathbf{x},\mu_j)$ with 
$$\tilde{{\Phi}}(\mathbf{x},\mu_j)\equiv\frac{1}{2\pi}\int \left. {\frac{\partial \Phi(\mathbf{x}-\bar{\bm{\rho}}'_0(\mathbf{x},\mu,\theta_1),\mu)}{\partial \mu}} \right |_{\mu=\mu_j}d\theta_1,$$
$$n_{i1}(\mathbf{x},\mu_j)\equiv \iint F_{i1j}(\mathbf{x}-\bar{\bm{\rho}}'_0(\mathbf{x},\mu_j,\theta_1),\mu_j,u_1)du_1 d\theta_1, $$
Then, $\tilde{{\Phi}}(\mathbf{x})$ and $n_{i1}(\mathbf{x})$  are obtained by the discrete sums
\begin{subequations}\label{e69}
\begin{eqnarray}
\tilde{{\Phi}}(\mathbf{x})&=&\sum \limits_{j=1}^{N_\mu} \tilde{{\Phi}}(\mathbf{x},\mu_j) \exp(\frac{-\mu_j B}{T_i}) B \delta\mu_j,  \label{e69.1}\\
n_{i1}(\mathbf{x})&=&\sum \limits_{j=1}^{N_\mu} n_{i1}(\mathbf{x},\mu_j) B\delta\mu_j.  \label{e69.2} 
\end{eqnarray}
\end{subequations}

In the standard model,  $\tilde{\Phi}(\mathbf{x},\mu_j)$  is replaced by  
$${\tilde\phi} (\mathbf{x},\mu_j)\equiv \frac{1}{2\pi}\int\Phi\big(\mathbf{X}-\bar{\bm{\rho}}'_0\big(\mathbf{x},\mu_j,\theta_1\big),\mu_j\big)d\theta_1,$$
and 
\begin{equation}\label{e71}
{\tilde{\phi}}(\mathbf{x})=\sum \limits_{j=1}^{N}{ \tilde{\phi}}(\mathbf{x},\mu_j) \exp(\frac{-\mu_j B}{T_i}) B\delta\mu_j. 
\end{equation}


\subsubsection{Interpolation algorithm to compute the gyroaverage and double-gyroaverage term}

To compute the gyroaverage and double-gyroaverage term, instead of truncating the Taylor expansion of the gyroaverage term at the second order, we implemented the interpolation algorithm, which replaces the integral of gyroaverage by a discrete sum of the function quantities over the Larmor circle and the function quantity at a point on the Larmor circle is obtained by the interpolation with cubic spline as an example. Due to that the number of interpolation points around the Larmor circle can be chosen arbitrarily, the integral of gyroaverage can be approximated with any accuracy by this interpolation method  by choosing enough interpolation points. Therefore, this numerical method can recover the short-scale information embodied by DGT theoretically, with only the constraint coming from the length scale of the mesh of the simulated domain. Since the interpolation coefficients only involves the equilibrium quantities, these coefficients can be assembled as a matrix and computed and stored at the beginning of simulations, preparing for the subsequent revoking \cite{refId0, Steiner2015}. 

To do this, we consider a uniform polar mesh on the domain $[r_{min},r_{max}]\times [0,2\pi]$ including $N_r \times N_{\Theta}$ cells:
\begin{eqnarray}
C_{h,j}=[r_h,r_{h+1}]\times[\Theta_p,\Theta_{p+1}], h=0,\cdots,N_r; p=0,\cdots,N_{\Theta}-1   \nonumber
\end{eqnarray}
where
\begin{eqnarray}
r_h &=& r_{min}+h\frac{r_{max}-r_{min}}{N_r}, h=0,\cdots,N_r \nonumber \\
\Theta_p &=& \frac{2\pi p}{N_{\Theta}}, p=0,\cdots,N_\Theta. \nonumber
\end{eqnarray}
The gyroangle $\theta\in [0,2\pi)$ is divided into $N_\theta$ equal segments with 
$$\theta_l=\frac{2\pi l}{N_\theta}, l\in\{0,\cdots, N_\theta -1\}.$$
The domain of magnetic moment $(0, \mu_{\max})$ is also divided into $N_\mu$ cells.  

The computation of $\Phi({r_h},{\Theta _p},x_\parallel,\mu_j)$ at a point $(r_h,\Theta_p)$ of the polar mesh as the first gyroaverage of  $\phi$  is approximated  by following discrete sum:
\begin{equation}\label{b45}
\Phi({r_h},{\Theta _p},x_\parallel,\mu_j) = \frac{1}{N}\sum\limits_{l= 0}^{N - 1} {\phi \left( {{r_h}\cos {\Theta _p} + \rho_j \cos \left( {\frac{{2l\pi }}{N}} \right),{r_h}\sin {\Theta _p} + \rho_j \sin \left( {\frac{{2l\pi }}{N}} \right)},  x_\parallel \right)},
\end{equation}
where $\rho_j=\sqrt{2\mu_j}$.
The computation of the term ${\tilde{\phi}}(\mathbf{x},\mu)$ as the  second gyroaverage of $\phi$ is approximated as
\begin{equation}\label{b46}
{\tilde{\phi}} ({r_h},{\Theta _p},x_\parallel, \mu_j) = \frac{1}{N}\sum\limits_{l = 0}^{N - 1} {\Phi \left( {{r_h}\cos {\Theta _p} - \rho_j \cos \left( {\frac{{2l\pi }}{N}} \right),{r_h}\sin {\Theta _p} - \rho_j \sin \left( {\frac{{2l\pi }}{N}} \right)},x_\parallel,\mu_j \right)}.
\end{equation}
The respective symbols $+$ and $-$ in Eq.(\ref{b45}) and Eq.(\ref{b46}) should be paid attention. $\tilde{\Phi}(\mathbf{x},\mu)$ is computed by 
$$\tilde{\Phi}({r_h},{\Theta _p},x_\parallel, \mu_j)=\frac{{\tilde{\phi}} ({r_h},{\Theta _p},x_\parallel, \mu_j+d\mu)-{\tilde{\phi}} ({r_h},{\Theta _p},x_\parallel, \mu_j-d\mu)}{2d\mu} .$$
$\Phi({r_h},{\Theta _p},x_\parallel,\mu_j) $, ${\tilde{\phi}} ({r_h},{\Theta _p},x_\parallel, \mu_j)$ and $\tilde{\Phi}({r_h},{\Theta _p},x_\parallel, \mu_j)$ for all $h$s and $p$s can be assembled as the product between  the respective matrix and a vector defined as
$$\phi\equiv (\phi_{0,0},\cdots,\phi_{N_r,0},\phi_{0,1},\cdots,\phi_{N_r,1},\cdots,\phi_{0,N_{\Theta-1}},\cdots,\phi_{N_r,N_{\Theta-1}})^t,$$
where
$${\phi}_{h,p} \equiv \phi({r_h},{\Theta _p},x_\parallel).$$
${\Phi}({r_h},{\Theta _p},x_\parallel, \mu_j)$ is the first gyroaverage term. The  electric field 
$$\mathbf{E}({r_h},{\Theta _p},x_\parallel, \mu_j)\equiv -\nabla {\Phi}({r_h},{\Theta _p},x_\parallel, \mu_j)$$
is used to drive the advection of $F_j(C,\mu_j,U)$ through Eq.(\ref{e68}).

Due to the periodic property in $\Theta$ dimension, the matrixes of ${\tilde{\phi}} ({r_h},{\Theta _p},x_\parallel, \mu_j)$ and $\tilde{\Phi}({r_h},{\Theta _p},x_\parallel, \mu_j)$ are of the circulant block structure, which in fourier basis can be transformed as block diagonal matrix. With FFT, their inverses can be easily solved. This technology is already used, for instance in Ref.\cite{refId0}. 

\subsubsection{The other algorithms used in the simulation}

The advection of the distribution uses the backward semi-Lagrangian scheme\cite{Sonnendrucker1999, Crouseilles2014, Grandgirard2006, Grandgirard2008, Latu2017}. The characteristics is given by Eq.(\ref{b34}). Since the Vlasov is written in a conservative from, it can be solved by splitting between the space and the velocity coordinates Ref.\cite{Latu2017, Grandgirard2006, Sonnendrucker1999, Garbet2010}. 

A. 1D advection along $x_\parallel$
$$(\partial_t+U \partial_ {x_\parallel})F(\mathbf{x},\mu_j,U)=0;$$

B. 1D advection along $U$
$$(\partial_t+\partial_{x_\parallel} \Phi {\partial_U})F(\mathbf{x},\mu_j,U)=0;$$

C. 2D advection in the cross section
$$({\partial_t}+\frac{1}{r}\partial_\Theta \Phi {\partial_r}-\frac{1}{r}\partial_r \Phi {\partial_\Theta})F(\mathbf{x},\mu_j,U)=0 .$$
The Verlet algorithm is used to find out the starting phase-space point of the characteristics ending at the mesh points. The two-dimensional cubic spline interpolation with periodic boundary condition on the polar angle dimension and natural boundary condition on the radial dimension  and $5$th order Lagrangian interpolation are used to compute the value of the distribution function at that starting point, which will be treated as the value of the distribution function at the associated mesh grid and as the initial value for the next iteration. 

\subsection{The initial distribution}

In the cylindrical coordinates system, the initial distribution is of the structure in Eq.(\ref{qq22}) and its specific formula is 
\begin{equation}\label{b39}
F(0,r,{x_\parallel },\mu, U,\Theta ) = {F_{eq}}\left( {r,\mu,U} \right) \times \left( {1 + \eta \exp \left( { - \frac{{{{(r - {r_p})}^2}}}{{\delta r}}} \right)\sum\limits_{n,m,l} {\cos (\frac{{2\pi n}}{L_\parallel}x_\parallel + m\Theta +\frac{2\pi p}{L_r} r)} } \right),
\end{equation}
where $n,m,p$ are the mode numbers in the respective dimensions and
the equilibrium function $F_{eq}$ is
\begin{equation}\label{b40}
F_{eq}\left( {r,\mu,U} \right) = \frac{n_0\left( r \right)\exp \left( -\frac{U^2}{2T_i\left( r \right)}-\frac{\mu B}{T_i(r)}\right)}{\left( 2\pi T_i(r) \right)^{3/2}}.
\end{equation}
The profile $T_i(r)$, $T_e(r)$ and $n_0(r)$ are given by:
\begin{equation}\label{b41}
\mathcal{P}(r) = {C_\mathcal{P}}\exp \left( { - {k_\mathcal{P}}\delta {r_\mathcal{P}}\tanh \left( {\frac{{r- {r_\mathcal{P}}}}{{\delta {r_\mathcal{P}}}}} \right)} \right)
\end{equation}
where $\mathcal{P}\in \{T_i,T_e,n_0\}$, $C_{T_i}=C_{T_e}=1$ and
\begin{equation}\label{b42}
{C_{{n_0}}} = \frac{{{r_{\max}} - {r_{\min}}}}{{\int_{{r_{\min}}}^{{r_{max}}} {\exp \left( { - {\kappa _{{n_0}}}\delta {r_{{n_0}}}\tanh \left( {\frac{{r - {r_P}}}{{\delta {r_{{n_0}}}}}} \right)} \right)dr} }}.
\end{equation}
We consider the parameters of  \cite{Coulette2013} [Medium case]:  $\eta=10^{-4}$, $k_{n_0}=13.2$, $\kappa_{T_i}=\kappa_{T_e}=66.0$, $\delta r_{T_e}=\delta r_{T_e}=0.1$, $L_\parallel=1500$, $r_p=0.5$, $\delta_r=0.2$. 
The simulation domain of $r\times\Theta \times x_\parallel \times U \times \mu$ is $(0.1, 14.5)\times [0,2\pi)\times (0, 1500)\times (-7.32,7.32)\times (0,7)$.

\subsection{Parallelization}

The simulation domain of $r\times\Theta \times x_\parallel \times U \times \mu$ is divided into the mesh with $128\times 64 \times 32 \times 32 \times 16$ cells.  The simulation is carried out on ATLAS4 of IRMA.  MPI is used in the parallelisation. 128 processors are divided into $16$ sub-communicators. $F_{j}(\mathbf{x},\mu_{j},U)$ with $j\in\{1,\cdots,16\}$ is exclusively computed by the $j$th sub-communicator. And the respective precomputing matrixes of $\Phi({r_h},{\Theta _p},x_\parallel,\mu_j)$, $\tilde{{\phi}}({r_h},{\Theta _p},x_\parallel,\mu_j)$ and ${\tilde\Phi}({r_h},{\Theta _p},x_\parallel,\mu_j)$ are stored in the  $j$th sub-communicator. 
$\tilde{{\Phi}}({r_h},{\Theta _p},x_\parallel)$, $n_{i1}({r_h},{\Theta _p},x_\parallel)$, $n_{i0}({r_h},{\Theta _p},x_\parallel)$ and ${\tilde{\phi}}({r_h},{\Theta _p},x_\parallel)$ in Eqs.(\ref{e69.1}-\ref{e69.2}) and (\ref{e71}) are computed by "MPI{\_}ALLREDUCE" the respective quantity stored in the processors of the same ``color'' with respect to the respective sub-communicator. 

To calculate the advection of distribution function in the 4D domain $(r,\Theta, x_\parallel,U)$,  two parallelization schemes are involved: the one of parallelizing $x_\parallel$ with $r,\Theta,U$ sequential is utilized to calculate the advection due to $\dot{r},\dot{\Theta},\dot{U}$; the other one of parallelizing $r,\Theta,U$ with $x_\parallel$ sequential is implemented to compute the advection due to  $U$. To compute the original points of the characteristic $\dot{r},\dot{\Theta},\dot{U}$, the parallelization of $x_\parallel$ with $r,\Theta,U$ sequential  is used. 
The parallelization in $x_\parallel$ with $r,\Theta$ sequential is implemented to compute QNE in the poloidal cross section. 

\subsection{The simulation results}

$\delta t=8$ is chosen as the time step in the simulations. $600$ steps are carried out and the data is stored every three steps.
The evolution of the potential profile on the polar cross section for both models is shown in Fig.(\ref{potentialevolution}). Both simulations begin with the same equilibrium density profile and the perturbative density profile. The potential profiles on the polar cross section at time moments $24, 4320$ computed by the two models are given in
Fig.(\ref{potentialevolution}). The evolution of the polar Fourier modes with the mode numbers $0,4,5,8,10,15$ of the potential are plotted in Fig.(\ref{polarmode}). Both model exhibit strong nonlinear interaction. The growth rate of polar mode $l=5$ at the radial grid $60$ of the two models is plotted in Fig.(\ref{growthrate}). The samplings of the radial Fourier spectrum of  the potential are plotted in  Fig.(\ref{radialpotential}). In Fig.(\ref{radialpotential}), the obvious difference between the two spectrums appears for the waves whose model numbers larger or equal $16$, indicating that the microturbulences computed by the two modes are different. Fig.(\ref{therm_diag}) plots the evolution of the quantity $\int_{r_{\min}}^{r_{\max}}\int_0^{2\pi}|\phi(r,\Theta,0)|^2rdrd\Theta$ computed by the two models. 


\begin{figure}[htbp]
\centering
 \includegraphics[width=0.35\textwidth, trim={1.5cm 0.5cm 2.5cm 0.5cm},clip]{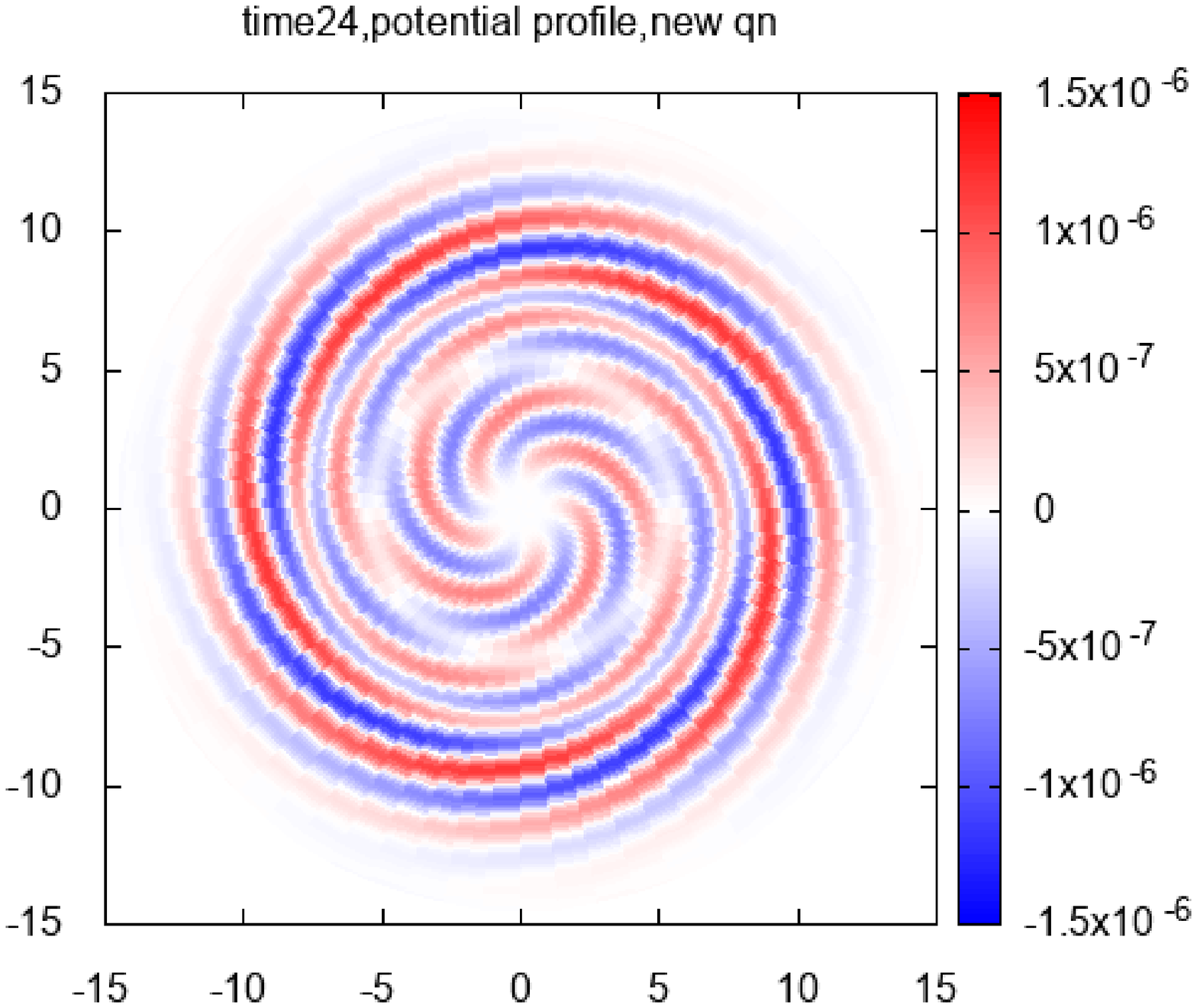}
 \includegraphics[width=0.35\textwidth, trim={1.5cm 0.5cm 2.5cm 0.5cm},clip]{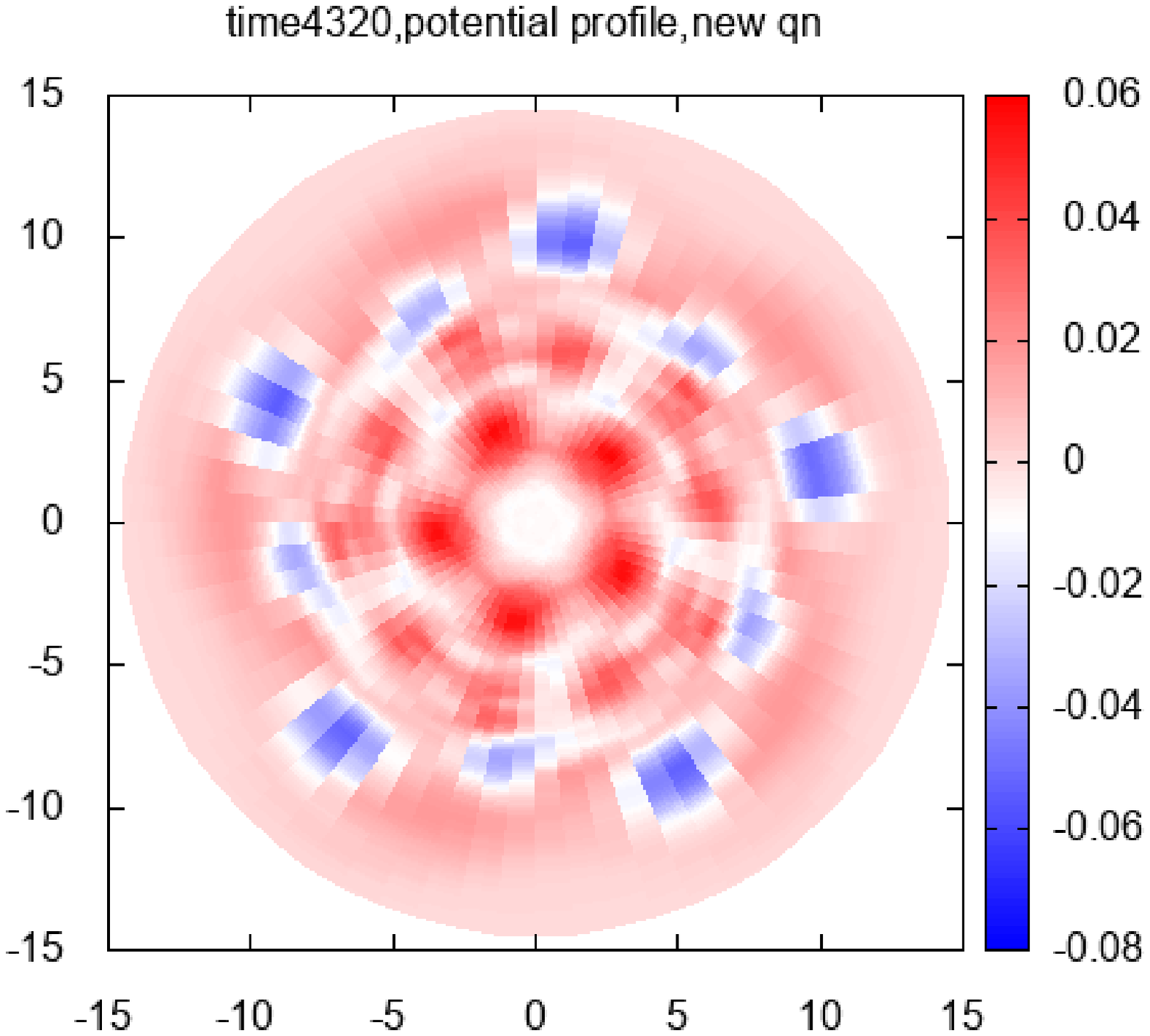}
\hspace{8pt}
\includegraphics[width=0.35\textwidth, trim={1.5cm 0.5cm 2.5cm 0.5cm},clip]{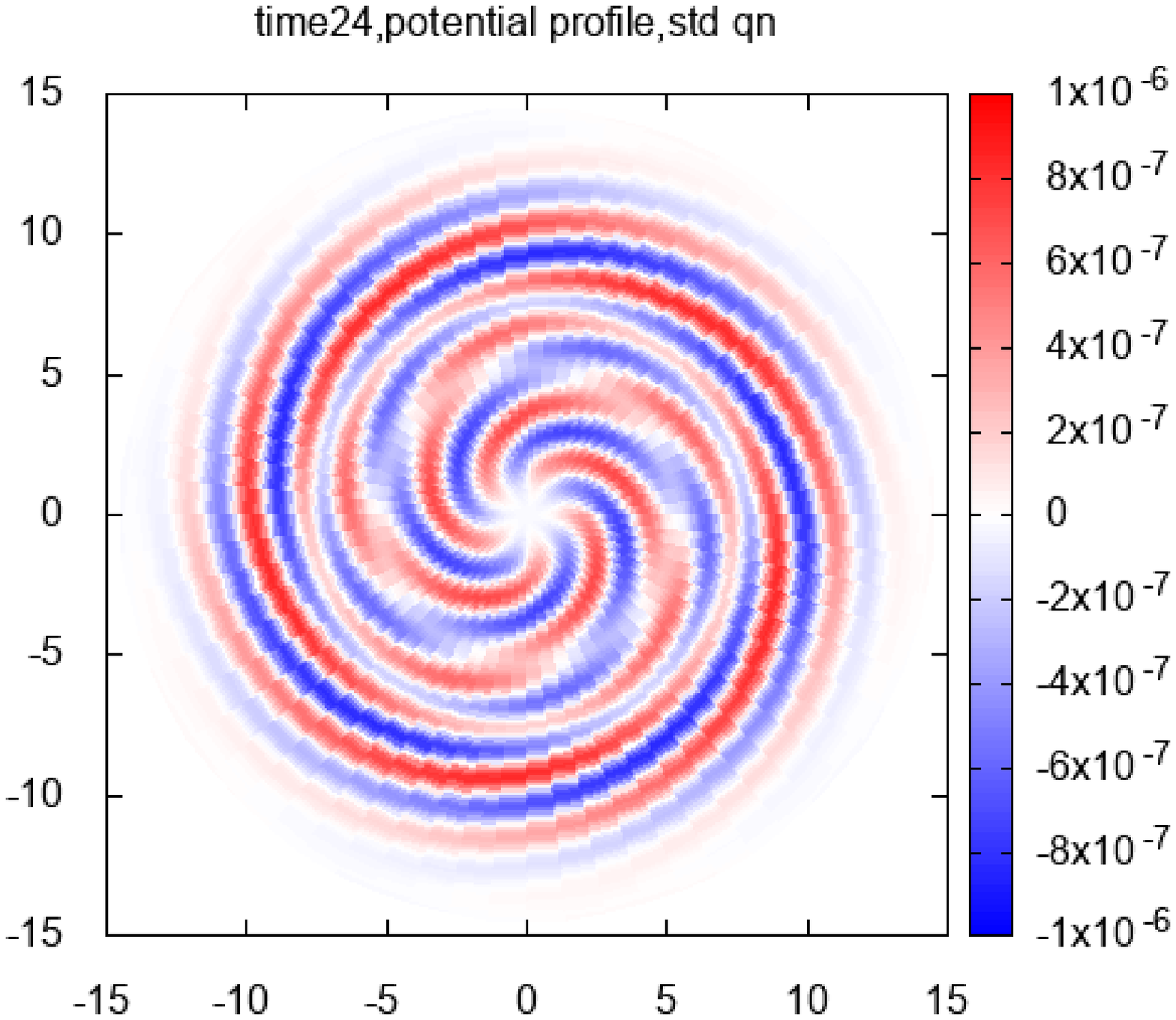}
\includegraphics[width=0.35\textwidth, trim={1.5cm 0.5cm 2.5cm 0.5cm},clip]{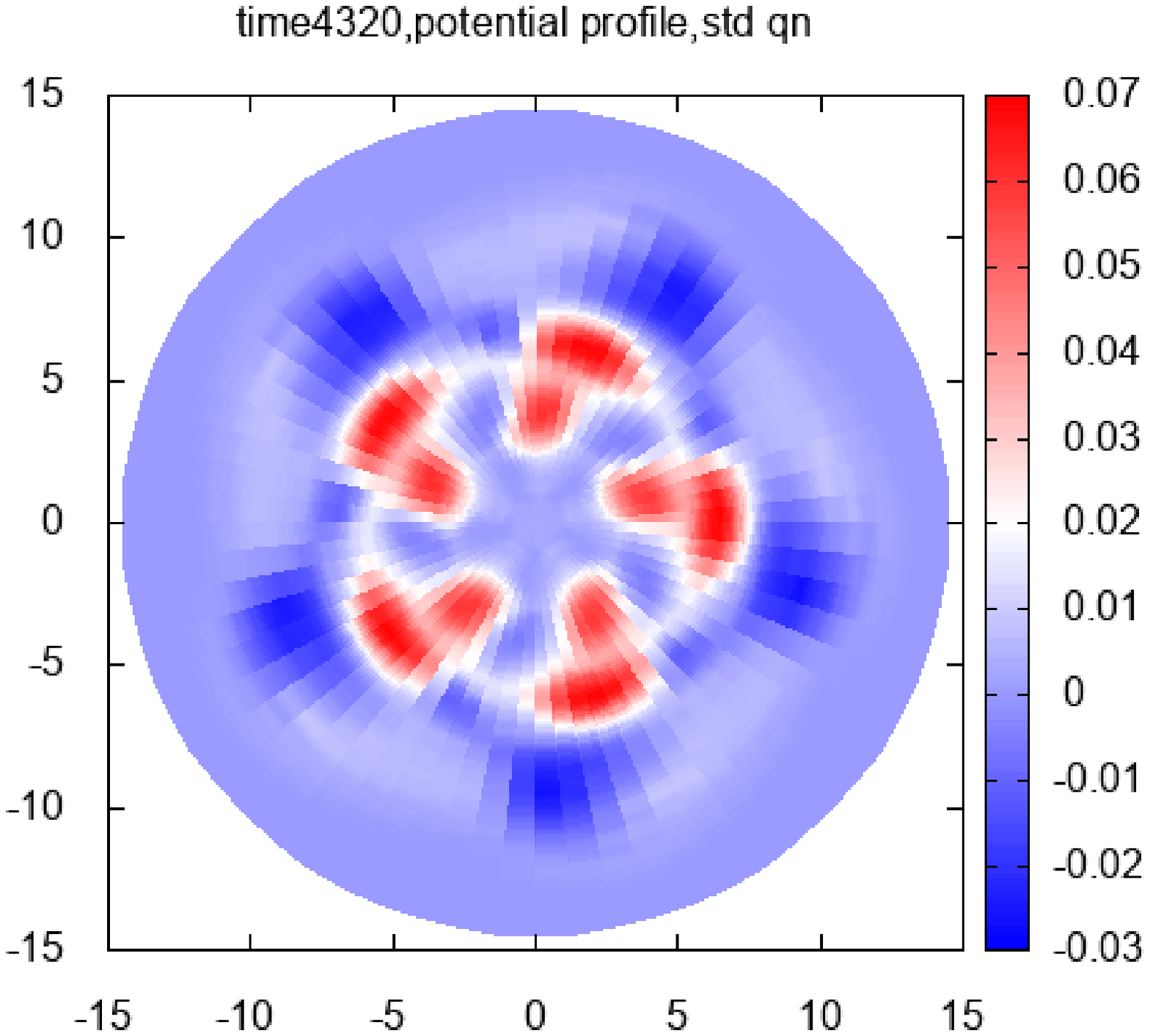}
\caption{\label{potentialevolution} The comparison of the potential profile on the polar cross section computed from the respective models at $t=24,4320$.  }
\end{figure}

\begin{figure}[htbp]
\centering
\subfloat[][]{\label{} \includegraphics[width=0.4\textwidth, trim=0.5 0.5 0 0.9]{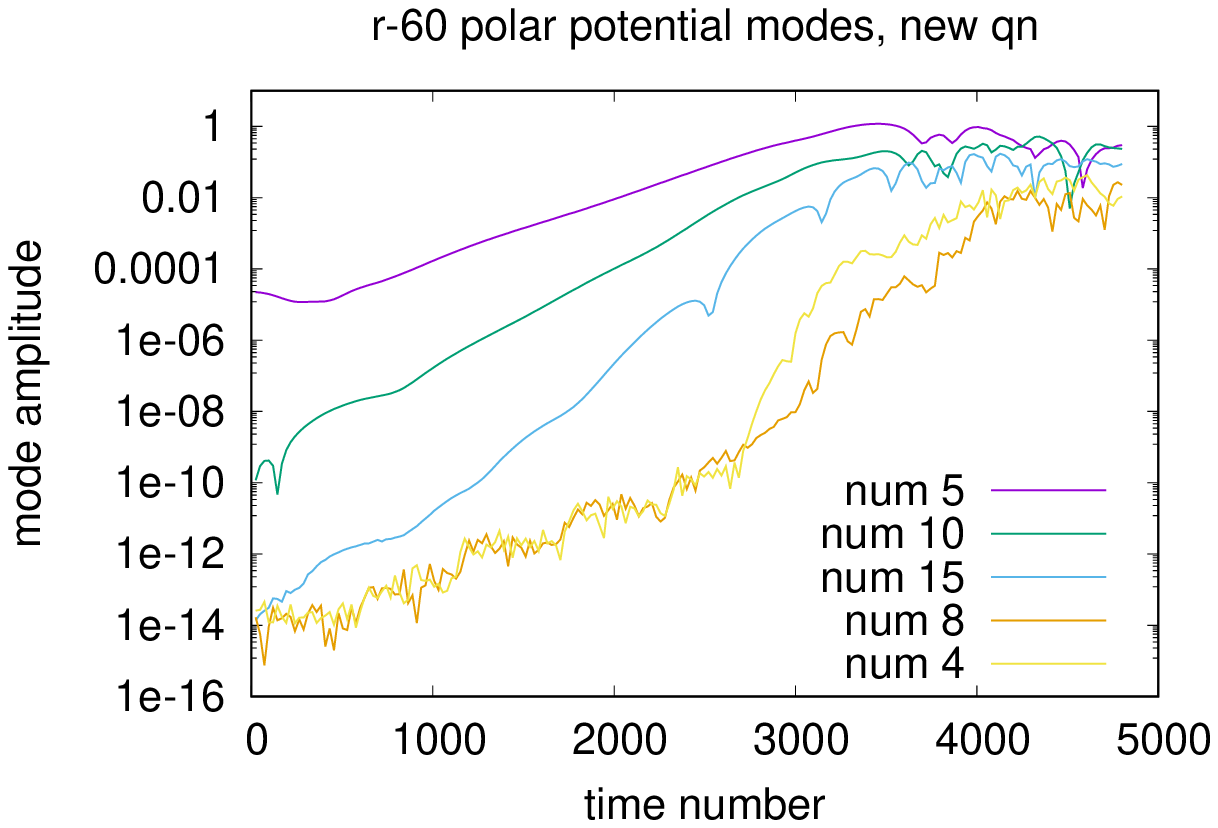}}
\subfloat[][]{\label{} \includegraphics[width=0.4\textwidth, trim=0.5 0.5 0 0.9]{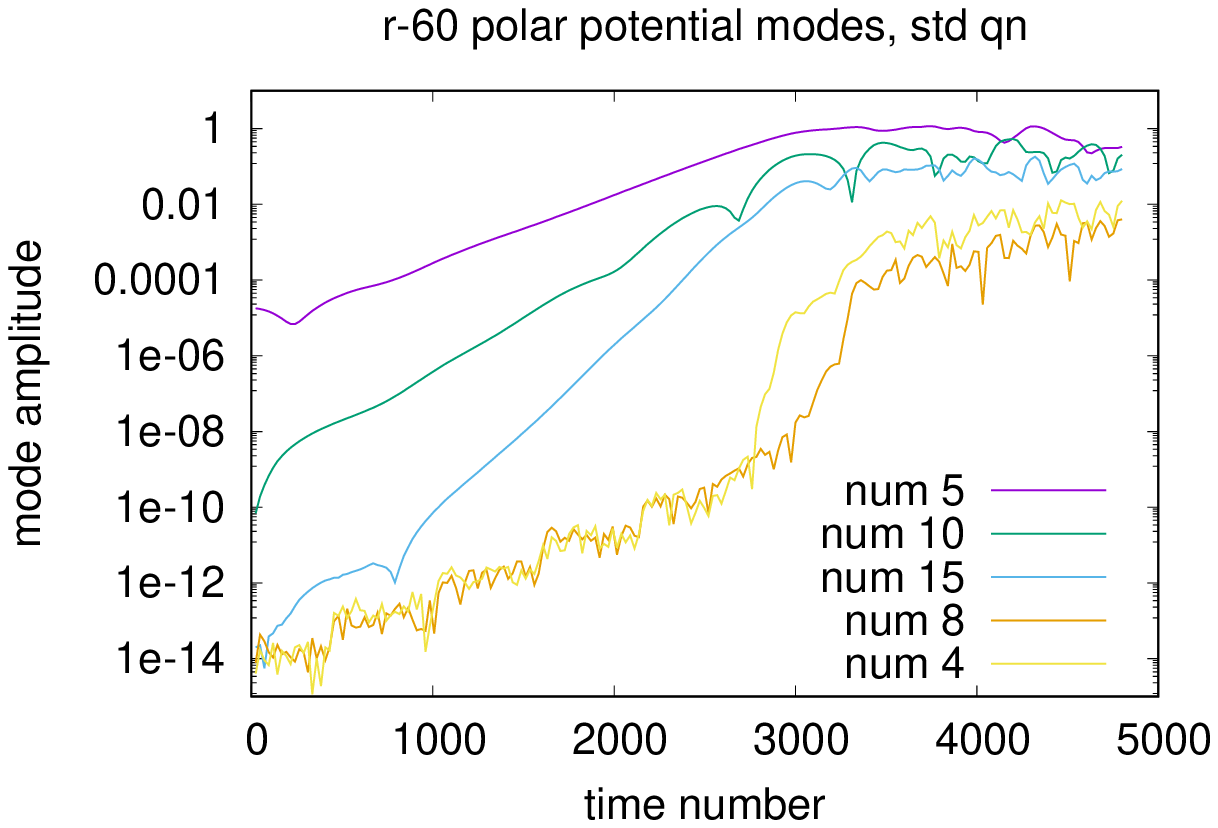}}
\caption{\label{polarmode} The evolution of the polar Fourier modes $0,4,5,8,10,15$ of the perturbative potential at radial node $60$ computed from the two models. }
\end{figure}

\begin{figure}[htbp]
\centering
\centerline{\includegraphics[width=0.5\textwidth, angle=-0, trim=0.5 0.5 0.0 5.0]{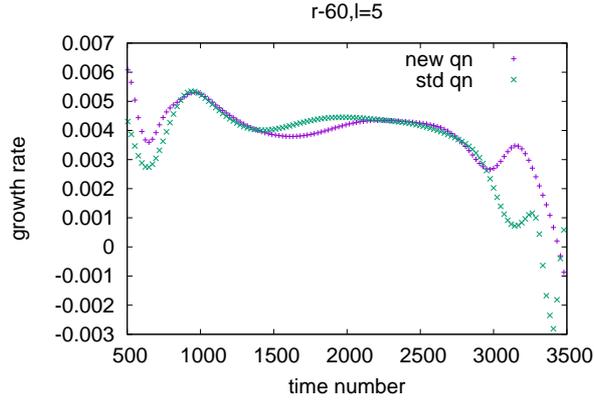}}
\caption{\label{growthrate}  The growth rate of main polar mode $l=5$ at the radial node $60$ computed from the two models. The saturation time of new model is later than that of the standard one.}
\end{figure}


\begin{figure}[htbp]
\centering
\subfloat[][]{\label{} \includegraphics[width=0.4\textwidth, trim=0.5 0.5 0 0.9]{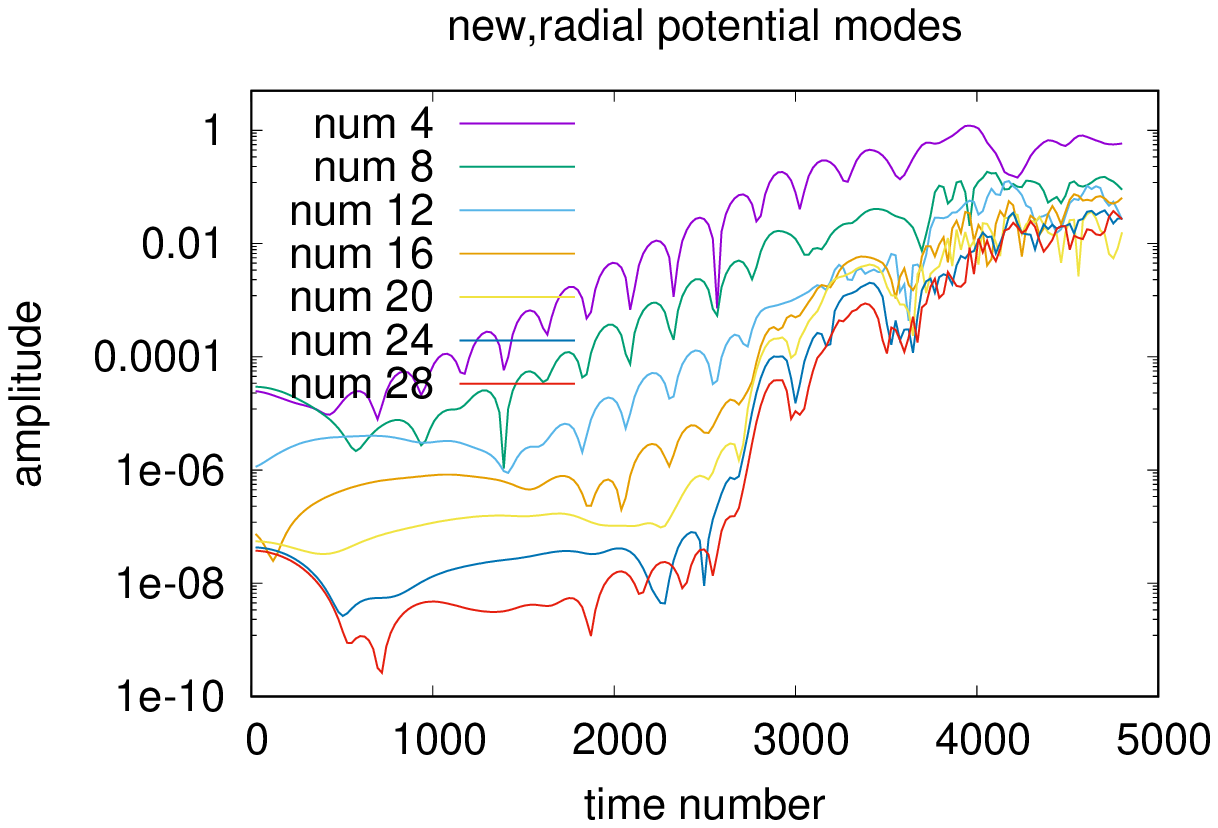}}
\subfloat[][]{\label{} \includegraphics[width=0.4\textwidth, trim=0.5 0.5 0 0.9]{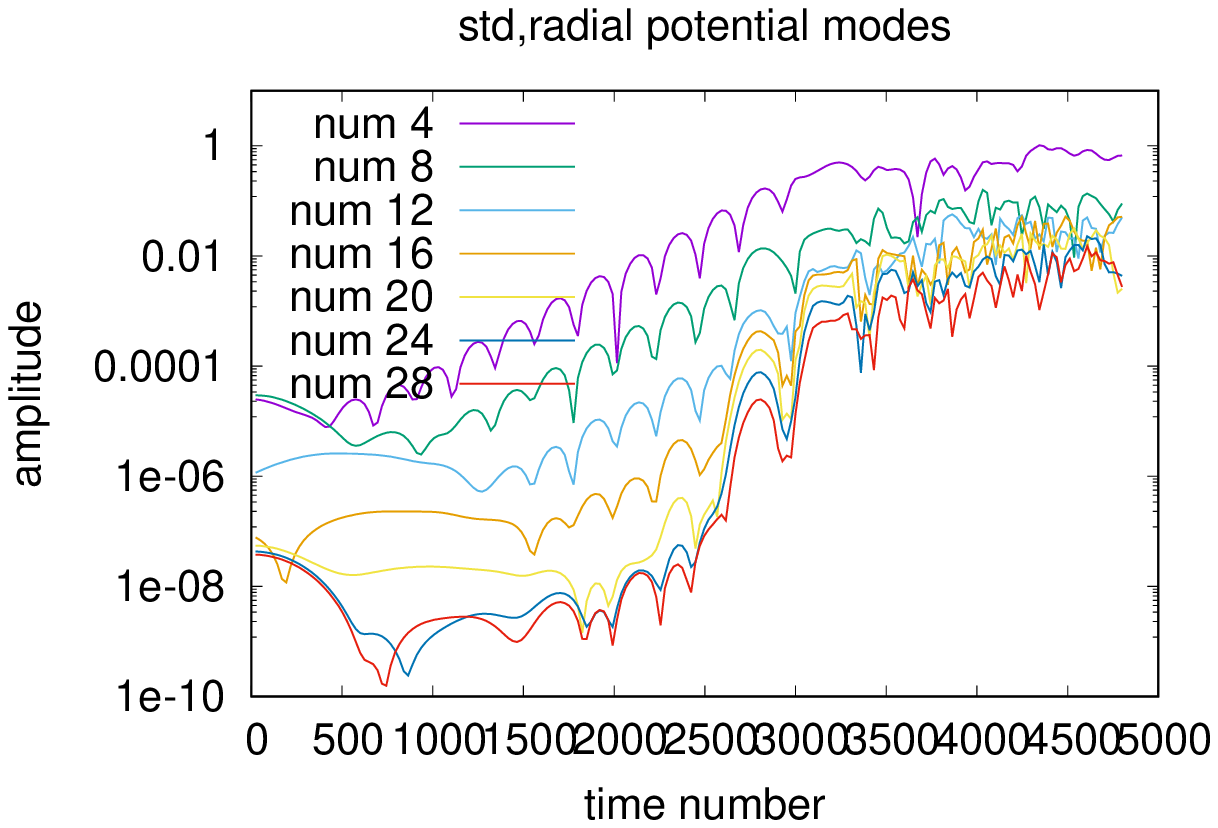}}
\caption{\label{radialpotential} The evolution of the radial Fourier spectrum of the perturbative potential computed by the two models.}
\end{figure}

\begin{figure}[htbp]
\centering
\centerline{\includegraphics[width=0.5\textwidth, angle=-0, trim=0.5 0.5 0.0 5.0]{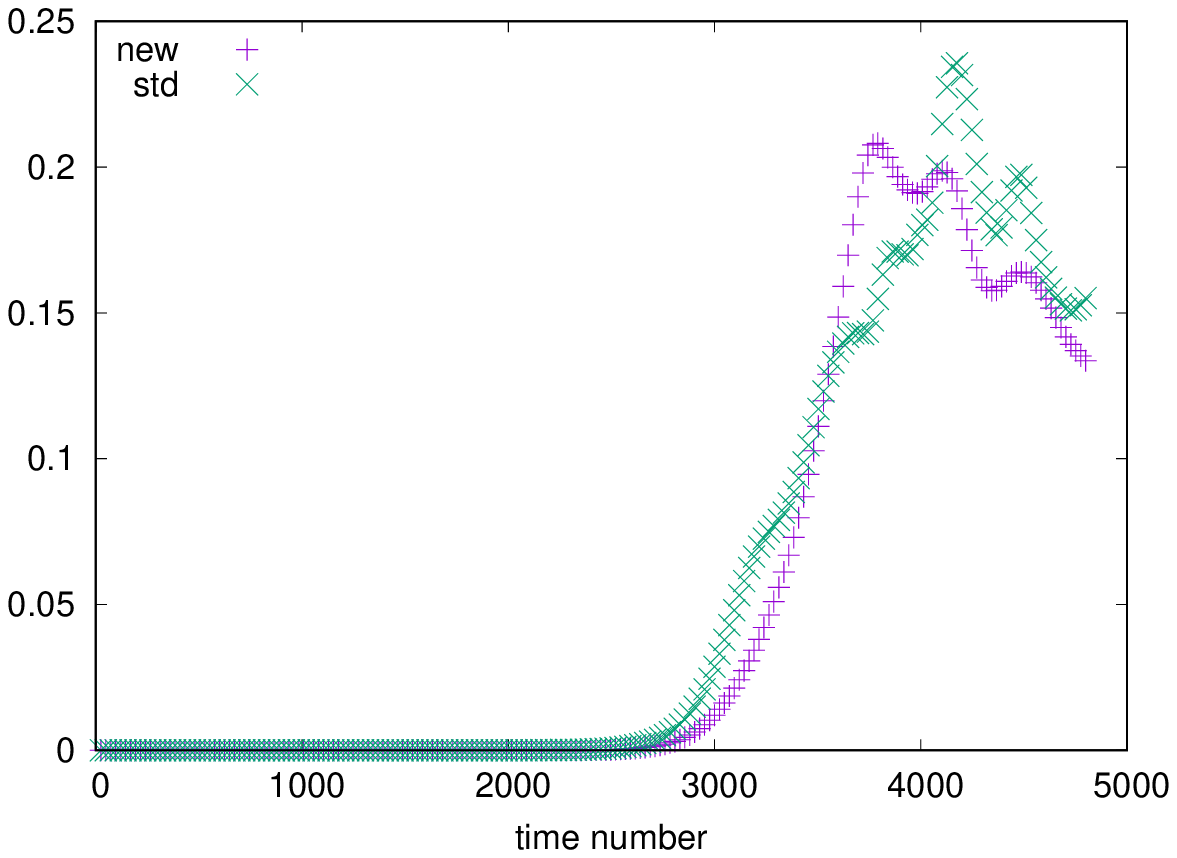}}
\caption{\label{therm_diag}   The time evolution of $\int_{r_{\min}}^{r_{\max}}\int_0^{2\pi}|\phi(r,\Theta,0)|^2rdrd\Theta$ computed by the two models.}
\end{figure}

\section{Acknowledgements}

The author thanks Prof. Michel Mehrenberger for the discussion of this work, and Prof. Phillippe Helluy for his funding supporting his work, and Dr. Sever Hirstoaga and Dr. Matthieu Boileau for the implementation of ATLAS HPC of IRMA.

\section{Summary and Discussion}

Through the order analysis, this paper pointed out that the full-orbit density derived by SGM is not truly accurate at the order $O(\varepsilon^{\sigma-1})$. By implementing a hybrid coordinate transform scheme, specifically, to transform the distribution on the gyrocenter coordinate to the one on the particle coordinate, we use the coordinate frame $(\mathbf{x},\mu,u_1,\theta_1)$ to replace $(\mathbf{x},\mu_1,u_1,\theta_1)$. The new full-orbit density derived by this method is truly accurate at the order $O(\varepsilon^{\sigma-1})$. The numerical simulations show that in the constant cylindrical magnetic field configuration, the two modes have the similar performance under the chosen plasma equilibrium profile.

\begin{appendices}

\section{ The coordinate transform derived  by Lie transform perturbative method, the equations of motion and SGM}\label{coordinates}


\subsection{The generators}\label{generator}

It's well-known that by the Euler-Lagrangian equations, the equations of motion can be derived by implementing the variational principle over the fundamental one-form. Gyrokinetic theory applies the Lie transform perturbative method to the fundamental one-form presented by Eq.(\ref{a11}) to obtain a new one independent of the gyroangle, through which the motion equations of other coordinates in a new version are independent of the gyroangle and the magnetic moment becomes a constant. Alternatively, it's a process to reduce the one dimension.  
The Lie transform perturbative method is introduced in Appendix. \ref{lie}.  The classical dimension-reduction process is divided into two steps\cite{Brizard1990}. The first step is to reduce the gyroangle from the non-perturbative one-form $\gamma_0+\varepsilon \gamma_1$ to get a non-perturbative one-form on guiding-center cordinates. Then, the perturbative potential is introduced into this new one form and the second-time Lie transform perturbative method is implemented to get a new one-form on gyrocenter coordinate and independent of the gyroangle. 

The generators of two consecutive transform are denoted by $\mathbf{g}_1$ and $\mathbf{g}_2$, where $\mathbf{g}_i\equiv(\mathbf{g}_i^\mathbf{X},g_i^\mu,g_i^U,g^{\theta}_i)$ for $i=1,2$ with $\mathbf{g}^\mathbf{X}_i\equiv (g_i^{1},g_i^2,g_i^3)$ being the spatial components. The subscript $i\in\{1,2\}$ are indexes for the guiding-center transform and gyrocenter transform, respectively.
 According to the classical method\cite{Brizard1990}, the first transform is only carried out to the second order of the exponential transform, while the second one is carried out  to the first order.  Specifically, the following equation
\begin{equation}\label{g4}
\bar{\Gamma} \left( \bar{\bf{Z}} \right)  =    [- {{L_{{{\bf{g}}_1}}}} +(L_{{{\bf{g}}_1}})^2](\gamma_0+\varepsilon \gamma_1) \left( \bar{\bf{Z}} \right),
\end{equation}
is to derive the new non-perturbative guiding-center fundamental one-form.
And to be consistent with transformations of one-form in Eqs.(\ref{g4}), the coordinate transforms is chosen as
\begin{equation}\label{c14}
\mathbf{z} =\bar{\mathbf{Z}}-\mathbf{g}_1+(\mathbf{g}_1\cdot \partial_{\bar{\mathbf{Z}}})^2\bar{\mathbf{Z}}.
\end{equation}
The second transform of the one-form is
\begin{equation}\label{c15}
\Gamma({\bf{Z}})  =   - {L_{{\bf{g}_2}}}[\bar{\Gamma}({\bf{Z}})+\varepsilon^\sigma \gamma_{\sigma}({\mathbf{Z}}-\mathbf{g}_1+(\mathbf{g}_1\cdot \partial_{{\mathbf{Z}}})^2{\mathbf{Z}})],
\end{equation}
and the associated coordinate transform is chosen as
\begin{equation}\label{c16}
\bar{\mathbf{Z}} = {\mathbf{Z}}-\mathbf{g}_2. 
\end{equation}

It's well-known \cite{Brizard1990, Brizard2007} that the generator for the guiding center is
$${\bf{g}}_1^{\bf{X}} =  - \varepsilon {\bm{\rho} _0}(\bar{\mathbf{Z}}). $$
with
\begin{equation}\label{g5}
{\bm{\rho} _0}(\bar{\mathbf{Z}})\equiv  \sqrt {\frac{{2\bar{\mu} }}{{B\left( {{\bar{\bf{X}}}} \right)}}} \left( { - {{\bf{e}}_1}\cos \bar{\theta}  + {{\bf{e}}_2}\sin \bar{\theta} } \right).
\end{equation}
Through Eq.(\ref{g4}), $\bar \Gamma \left( {\bf{Z}} \right)$ is 
\begin{equation}\label{g6}
\bar \Gamma \left( \bar{\bf{Z}} \right) = {\bf{A}}\left( {{\bf{\bar X}}} \right)\cdot d{\bf{\bar X}} + \varepsilon \bar U{\bf{b}}\cdot d{\bf{\bar X}}   +\varepsilon^2 \bar{\mu}d\theta - {\varepsilon \left( {\frac{{\bar U_{}^2}}{2} + \bar \mu B\left( {{\bf{\bar X}}} \right)} \right)  } dt +\mathscr{O}(\varepsilon^3).
\end{equation}
Here, $\mathscr{O}(\varepsilon^3)$ denotes that the coefficients of one-form $dX^j,d\mu,dU,d\theta,dt$ contained by the uncertain terms are of the order $O(\varepsilon^3)$ and this usage of ``$\mathscr{O}$'' to denote the order of the coefficients of the uncertain terms of the fundamental one-form will also be implemented in the following context. 

The exact order of $\bar \Gamma \left( \bar{\bf{Z}} \right)$ is $O(\varepsilon^2)$, while $O(\varepsilon^3)$ as the uncertain term will be ignored. 
 Now, substituting $\bar \Gamma \left( \bar{\bf{Z}} \right) $ into Eq.(\ref{c15}), ${\Gamma}(\mathbf{Z})$ can be separated into two parts. 
The first part is
\begin{equation}\label{g8}
{\Gamma _0}(\mathbf{Z})  = {\bf{A}}\left( {\bf{X}} \right)\cdot d{\bf{X}} + \varepsilon U{\bf{b}}\cdot d{\bf{X}} + {\varepsilon ^2}\mu d\theta  -  \left( {\varepsilon \mu B\left( {\bf{X}} \right) + \varepsilon \frac{{m{U^2}}}{2}}-\Gamma_{1t} \right)dt,
\end{equation}
while the second one being
\begin{eqnarray}\label{g7}
{\Gamma _1}(\mathbf{Z}) &=& \left( { - \left( {{\bf{B}} + \varepsilon U\nabla  \times {\bf{b}}} \right) \times {\bf{g}}_2^{\bf{X}} - \varepsilon g_2^U{\bf{b}}} \right)\cdot d{\bf{X}}  \nonumber \\
&& + \varepsilon \left( {{\bf{g}}_2^{\bf{X}}\cdot{\bf{b}}} \right)dU - {\varepsilon ^2}{g_2^\mu }d\theta  - {\varepsilon ^2}g_2^\theta d\mu  \nonumber \\
&& - [ \varepsilon \mu {\bf{g}}_2^{\bf{X}}\cdot\nabla B\left( {\bf{X}} \right) - \varepsilon Ug_2^U - \varepsilon g_2^\mu B + \varepsilon^\sigma \phi(\mathbf{X}+\varepsilon {\bm{\rho} _0(\mathbf{Z})})+\Gamma_{1t}] dt \nonumber \\
&& + dS_1+\mathscr{O}(\varepsilon^3)+\mathscr{O}(\varepsilon^{\sigma+1}),
\end{eqnarray}
where, $\Gamma_{1t}$ and $S_1$ will be solved. 
To get Eq.(\ref{g7}), the non-zero components of the Lie derivative on $\Gamma_0$ given by Appendix.\ref{app1} are used.
In Eq.(\ref{g7}), $\mathscr{O}(\varepsilon^3)$ is inherited from Eq.(\ref{g6}) and $\mathscr{O}(\varepsilon^{\sigma+1})$ is produced by approximating $\phi(\mathbf{X}+\varepsilon {\bm{\rho} _0(\mathbf{Z})}+(\mathbf{g}_1\cdot \partial_{{\mathbf{Z}}})^2 {\mathbf{Z}})$ as $\phi(\mathbf{X}+\varepsilon {\bm{\rho} _0(\mathbf{Z})})$ .  To make sure that ${\Gamma _1}(\mathbf{Z}) $ is exactly correct at the order $\varepsilon^\sigma$, alternatively, to make sure that $\varepsilon^\sigma \phi(\mathbf{X}+\varepsilon {\bm{\rho} _0(\mathbf{Z})})$ is the exact-order term, we require
\begin{equation}\label{c17}
\sigma < 3. 
\end{equation}

To remove the $\theta$-dependent terms in Eq.(\ref{g7}), the following identities are required
\begin{equation}\label{g9}
\Gamma_{1k}=0,Z^k\in \{\mathbf{X},U,\mu,\theta\}
\end{equation}
plus a requirement that $\Gamma_{1t}$ is independent of $\theta$. $S_1$ is the gauge function to be solved. Then, all the generators can be derived as
\begin{subequations}\label{g10}
\begin{eqnarray}
{\bf{g}}_2^{\bf{X}}&=&  - \frac{{{\bf{b}} \times \nabla {S_1}}}{{{\bf{b}}\cdot{{\bf{B}}^*}}} - \frac{{{{\bf{B}}^*}}}{\varepsilon }\frac{{\partial {S_1}}}{{\partial U}}, \\
g_2^U &=& \frac{1}{\varepsilon }{\bf{b}}\cdot\nabla {S_1}, \\
g_2^\mu & =& \frac{1}{{{\varepsilon ^2}}}\frac{{\partial {S_1}}}{{\partial \theta }}, \\
{g_2^\theta } & =&  - \frac{1}{{{\varepsilon ^2}}}\frac{{\partial {S_1}}}{{\partial \mu }},
\end{eqnarray}
\end{subequations}
with
$${{\bf{B}}^*} \equiv {\bf{B}} + \varepsilon U\nabla  \times {\bf{b}}.$$
The equation of the gauge function is
\begin{equation}\label{g14}
\frac{{\partial {S_1}}}{{\partial t}} + U{\bf{b}}\cdot\nabla {S_1} + \frac{B(\mathbf{X})}{\varepsilon }\frac{{\partial {S_1}}}{{\partial \theta }} = \varepsilon^\sigma \phi \left( {{\bf{X}} + \varepsilon {\bm{\rho} _0}}(\mathbf{Z}) \right) + {\Gamma _{1t}}.
\end{equation}

For the low frequency perturbation, inequalities $\left| {\frac{{\partial {S_1}}}{{\partial t}}} \right| \ll \left| {\frac{B}{{{\varepsilon }}}\frac{{\partial {S_1}}}{{\partial \theta }}} \right|, \left| {U{\bf{b}}\cdot\nabla {S_1}} \right| \ll \left| {\frac{B}{{{\varepsilon }}}\frac{{\partial {S_1}}}{{\partial \theta }}} \right|$ hold. By ignoring the two terms of higher order on the left of Eq.(\ref{g14}), the rest of Eq.(\ref{g14}) is
\begin{equation}\label{g15}
\frac{{B({\bf{X}})}}{\varepsilon }\frac{{\partial {S_1}}}{{\partial \theta }} =
\varepsilon^\sigma\phi ({\bf{X}} + \varepsilon {\bm{\rho} _0(\mathbf{Z})}) + {\Gamma _{1t}}.
\end{equation}
To remove the secularity of $S_1$ on the integration of $\theta$, $\Gamma_{1t}$ is chosen as
$${\Gamma _{1t}} =-\varepsilon^\sigma \Phi(\mathbf{X},\mu)$$
with the definition
\begin{equation}\label{ex1}
\Phi ({\bf{X}},\mu ) \equiv \left\langle {\phi \left( {{\bf{X}} + \varepsilon{\bm{\rho} _0}}(\mathbf{Z}) \right)} \right\rangle  = \frac{1}{{2\pi }}\int_0^{2\pi } {\phi \left( {\bf{x}} \right)\delta \left( {{\bf{x}} - {\bf{X}} - {\varepsilon\bm{\rho} _0(\mathbf{Z})}} \right)d\theta }
\end{equation}

The reason for removing the secularity from $S_1$ is that those secular terms could contribute unlimited terms to the generators through Eqs. (\ref{g10}). These unlimited terms cause the coordinate transform unacceptable. The solution of Eq.(\ref{g15}) is
\begin{equation}\label{ex2}
{S_1} = \frac{\varepsilon^{\sigma+1} }{{B\left( {\bf{X}} \right)}}\int_{}^\theta  \Psi  \left( {\bf{Z}} \right)d{\theta _1}+ \Pi \left( {{\bf{X}},\mu ,U} \right),
\end{equation}
with
\begin{equation}\label{ex3}
\Psi \left(\mathbf{Z} \right) \equiv \phi \left( {{\bf{X}} + \varepsilon {\bm{\rho} _0(\mathbf{Z})}} \right) - \Phi(\mathbf{X},\mu).
\end{equation}
 $\Pi \left( {{\bf{X}},\mu ,U} \right)$ is a function independent of $\theta$ and we choose it as zero here.

To get the order of the generators, we need the facts that $O(\varepsilon \mathcal{K}_\perp)= O(1)$ and $O(S_1)=O(\varepsilon^{\sigma+1})$. We also make the follow assumption that there doesn't exist large gradients in the $\mu$, $\theta$ and $U$ dimensions, so that $O(\left\|\partial_\theta \right \|)=O(\left\|\partial_\mu \right \|)=O(\left\|\partial_U \right \|)=O(1)$ holds. Then, the order of the four generators can be estimated as follows
\begin{subequations}\label{qq2}
\begin{eqnarray}
O(\left\| {{\bf{g}}_2^{\bf{X}}} \right\|) &=& O({\varepsilon ^\sigma })  \\
O(\left\| {g_2^U} \right\|) &=& O({\varepsilon ^\sigma })\\
O(\left\| {g_2^\mu } \right\|) &=& O({\varepsilon ^{\sigma  - 1}})\\
O(\left\| {g_2^\theta } \right\|) &=& O({\varepsilon ^{\sigma  - 1}})
\end{eqnarray}
\end{subequations}
Since the order of $g_2^\theta$ and $g_2^\mu$ is the same and lower than that of ${{\bf{g}}_2^{\bf{X}}}$ and ${g_2^U}$, only $g_2^\theta$ and $g_2^\mu$ are kept to participate in the coordinate transform between the full-orbit coordinate and the gyrocenter coordinate. By  making the following replacements
$$g_2^{\mu } \to {\varepsilon ^{\sigma  - 1}}{g_2^\mu }, \;\; g_2^{\theta } \to {\varepsilon ^{\sigma  - 1}}{g_2^\theta },$$
we have the order 
$$O(g_2^{\mu })=O(1), \;\; O(g_2^{\theta })=O(1). $$  
And  the solution of ${g}_2^{\mu}$ and $g_2^{\theta}$ with the arguments being $ ({\bf{X}},{\mu},{\theta} )$ are listed here:
\begin{subequations}\label{ee7}
\begin{eqnarray}
g_2^\mu \left( {{\bf{ X}},{\mu}},\theta \right) &=& \frac{\Psi \left( {{\bf{ X}},{\mu}},\theta \right)}{{ B\left( { {\bf{X}} } \right)}}, \\
g_2^\theta \left({{\bf{X}},{\mu}},\theta \right) &=&   \frac{{\partial _{{\mu}} }\int_0^{{\theta}} {\Psi \left({{\bf{ X}},{\mu}}, \theta \right)d{\theta}}}{{B\left( { {\bf{X}} } \right)}}
\end{eqnarray}
\end{subequations}
with
\begin{equation}\label{ee8}
\Psi \left({\mathbf{X}},{\mu}\right) \equiv \phi \left( {{\bf{X}} +  \varepsilon {\bm{\rho} _0}} \right) - \Phi({\mathbf{X}},{\mu}).
\end{equation}
\begin{equation}\label{ee9}
\Phi ({\bf{ X}}, \mu ) \equiv \left\langle {\phi \left( {{\bf{ X}} + {\varepsilon\bm{\rho} _0}} \right)} \right\rangle  = \frac{1}{{2\pi }}\int_0^{2\pi } {\phi \left( {{\bf{ X}} + {\varepsilon\bm{\rho} _0}} \right)d\theta } .
\end{equation}

\subsection{The coordinates transform}

According to Eq.(\ref{c16}),  the transform from the guiding-center coordinate to the gyrocenter coordinate $\psi _{gy}:{\bar {\bf{Z}}} \to {\bf{Z}}$ is approximated with the exact order $O(\varepsilon^{\sigma-1})$ 
\begin{subequations} 
\begin{eqnarray}
{\bf{\bar X}} &=&  {\bf{X}} ,  \nonumber \\
 \bar \mu &=& \mu  - \varepsilon^{\sigma-1}g_2^\mu \left( {{\bf{ X}},{\mu}},\theta \right), \nonumber \\
\bar U &=& U, \nonumber \\
 \bar \theta&=&\theta    - \varepsilon^{\sigma-1}g_2^\theta \left({{\bf{ X}},{\mu}},\theta \right), \nonumber 
\end{eqnarray}
\end{subequations}
which can be rearranged with the exact order being $O(\varepsilon^{\sigma-1})$  as
\begin{subequations}\label{ee3}
\begin{eqnarray}
{\bf{X}} &=& {\bf{\bar X}} , \label{ee3.1} \\
\mu  &=& \bar \mu  +\varepsilon^{\sigma-1}g_2^\mu \left( {{\bf{\bar X}},\bar{\mu}},\bar{\theta} \right)+O(\varepsilon^{2\sigma-2}), \label{ee3.2} \\
U &=& \bar U, \label{ee3.3} \\
\theta  &=& \bar \theta  +\varepsilon^{\sigma-1}g_2^\theta \left({{\bf{\bar X}},\bar{\mu}},\bar{\theta} \right)+O(\varepsilon^{2\sigma-2}).\label{ee3.4} 
\end{eqnarray}
\end{subequations}

While based on Eq.(\ref{c14}), the coordinate transform from the full orbit to the guiding-center coordinate $\psi_{gc}: \mathbf{z}\to \mathbf{\bar{Z}}$ is approximated  exactly right at  $O(\varepsilon)$
\begin{subequations}
\begin{eqnarray}
 {\bf{x}} &=& {\bf{\bar X}} + \varepsilon{\bm{\rho} _0}(\bar{\mathbf{X}},\bar{\mu}_1,\bar{\theta}_1)+O(\varepsilon^2),  \nonumber \\
{\mu _1} &=& \bar \mu  , \nonumber \\
 {u_1}&=& \bar U , \nonumber \\
{\theta _1}&=& \bar \theta ,  \nonumber
\end{eqnarray}
\end{subequations}
which can also be rearranged  exactly right at  $O(\varepsilon)$ as
\begin{subequations}\label{ee2}
\begin{eqnarray}
{\bf{\bar X}}  &=& {\bf{x}}  - \varepsilon{\bm{\rho} _0}(\mathbf{x},\mu_1,\theta_1) +O(\varepsilon^2),  \label{ee2.1}\\
\bar \mu  &=& {\mu _1} , \label{ee2.2}\\
\bar U &=& {u_1}, \label{ee2.3}\\
\bar \theta &=& {\theta _1}. \label{ee2.4}
\end{eqnarray}
\end{subequations}

\subsection{The equations of motion}

The new fundmental one-form with exact order $O(\varepsilon^2)$ and uncertain order $O(\varepsilon^3)$ is
\begin{equation}\label{g17}
\begin{split}
\Gamma  &= \left( {{\bf{A}}\left( {\bf{X}} \right) + \varepsilon U{\bf{b}}} \right)\cdot d{\bf{X}} + {\varepsilon ^2}\mu d\theta  \\
& - \left( {\varepsilon \left( {\mu B\left( {\bf{X}} \right)}  + \frac{{{U^2}}}{2} \right)+\varepsilon^\sigma \Phi(\mathbf{X},\mu)} \right)dt +\mathscr{O}(\varepsilon^3),
\end{split}
\end{equation}
which is exactly right at $\mathscr{O}(\varepsilon^\sigma)$.
The Lagrangian derived from Eq.(\ref{g17}) is
\begin{equation}\label{gg18}
\begin{split}
\mathcal{L} &= \left( {{\bf{A}}\left( {\bf{X}} \right) + \varepsilon U{\bf{b}}} \right)\cdot \dot{\bf{X}} + {\varepsilon ^2}\mu \dot{\theta}   \\
& - \left( {\varepsilon \left( {\mu B\left( {\bf{X}} \right)}  + \frac{{{U^2}}}{2} \right)+\varepsilon^\sigma \Phi(\mathbf{X},\mu)} \right) +{O}(\varepsilon^3)
\end{split}
\end{equation}
Applying the variational principle to this Lagrangian 1-form given by Eq.(\ref{g17}), the orbit equations are derived exactly right at $O(\varepsilon^{\sigma-1})$
\begin{subequations}\label{g18}
\begin{eqnarray}
\mathop {\bf{X}}\limits^.  &=& \frac{{U{{\bf{B}}^*} + {\bf{b}} \times \nabla \left( {\varepsilon \mu B\left( {\bf{X}} \right) + \varepsilon^\sigma \Phi ({\bf{X}},\mu )} \right)}}{{{\bf{b}}\cdot{{\bf{B}}^*}}} +O(\varepsilon^3), \label{g18a} \\
\dot U &=& \frac{{ - {{\bf{B}}^*}\cdot\nabla \left( {\varepsilon \mu B\left( {\bf{X}} \right) + \varepsilon^\sigma \Phi ({\bf{X}},\mu )} \right)}}{{\varepsilon {\bf{b}}\cdot{{\bf{B}}^*}}}+O(\varepsilon^2) \label{g18b}
\end{eqnarray}
\end{subequations}
where $\mathbf{B}^*(\mathbf{X})\equiv \mathbf{B}(\mathbf{X})+\varepsilon U\nabla \times \mathbf{b}$.
Eq.(\ref{e57}) is used to obtain the exact order in Eqs.(\ref{g18a},\ref{g18b}).

\subsection{The transform of the distribution}\label{error.2}

For the Vlasov gyrokinetic simulation, we need to transform the distribution function from the gyrocenter coordinate to the full-orbit coordinate\cite{Garbet2010}. With the coordinate transform composited by Eqs.(\ref{ee2}) and (\ref{ee3}),  given a distribution function on the gyrocenter coordinate $F_s\left( {{\bf{X}},\mu ,U},t \right)$, the distribution function on the full orbit can be derived by following the transform chain
\begin{equation}\label{qq1}
F_s\left  ( {{\bf{X}},\mu ,U}\right)\mathop{\longrightarrow} \limits^{\psi _{gy}} \bar{F}_s\left(\bar{\mathbf{Z}} \right)\mathop{\longrightarrow} \limits^{\psi _{gc}} f_s\left(\mathbf{z} \right).
\end{equation}
First, the  total distribution function is separated into the sum of an equilibrium one plus a perturbative one as
\begin{equation}\label{qq3}
F_s\left( \mathbf{X},\mu,U \right) = {F_{s0}}\left(  \mathbf{X},\mu,U \right) + {F_{s1}}\left(  \mathbf{X},\mu,U \right).
\end{equation}

\begin{prop}\label{propf2.1}
By dividing $F_s$ as Eq.(\ref{qq3}) does, the exact order of $F_{s1}$ equals $O(\varepsilon^{\sigma-1})$ with respect to the low frequency perturbations, specifically, $O(||\partial_t||_1)=O(1)$, where the subscript $1$ denotes the operation on the perturbative quantity. .
\end{prop}

\begin{proof}
The  Vlasov equation $\frac{dF_s}{dt}=(\partial_t+\dot{\mathbf{X}}\cdot\nabla+\dot{U}\partial_U)F_s=0$ can be linearized as the sum of two parts depending on $F_{s0}$ and $F_{s1}$, respectively
\begin{equation}\label{f2}
\big (\frac{dF_s}{dt}\big)_{P}+\big(\frac{dF_s}{dt}\big)_{E} =0,
\end{equation}
\begin{subequations}
\begin{eqnarray}
&&\big(\frac{dF_s}{dt}\big)_{P}\equiv (\dot{\mathbf{X}}_{P}\cdot \nabla +\dot{U}_P \partial_U) F_{s0}, \nonumber  \\
&&\big(\frac{dF_s}{dt}\big)_{E}\equiv (\partial_t +\dot{\mathbf{X}}_{E}\cdot \nabla +\dot{U}_E \partial_U) F_{s1}, \nonumber
\end{eqnarray}
\end{subequations}
with 
\begin{subequations}
\begin{eqnarray}
\dot{\mathbf{X}}_{E} &\equiv& \frac{{U{{\bf{B}}^*} + {\bf{b}} \times \nabla \left( {\varepsilon \mu B\left( {\bf{X}} \right)} \right)}}{{{\bf{b}}\cdot{{\bf{B}}^*}}}, \nonumber \\
\dot{\mathbf{X}}_{P}&\equiv & \frac{ {\bf{b}} \times \nabla \left( \varepsilon^\sigma \Phi ({\bf{X}},\mu ) \right) }
{{\bf{b}}\cdot{{\bf{B}}^*}}, \nonumber \\
\dot{U}_E &\equiv &\frac{- {{\bf{B}}^*}\cdot \nabla\mu B\left( {\bf{X}} \right) }{{ {\bf{b}}\cdot{{\bf{B}}^*}}}, \nonumber \\
\dot{U}_P &\equiv &\frac{- {{\bf{B}}^*}\cdot \nabla\varepsilon^\sigma \Phi ({\bf{X}},\mu ) }{{ \varepsilon {\bf{b}}\cdot{{\bf{B}}^*}}}, \nonumber 
\end{eqnarray} 
\end{subequations}
where the equations of motion are derived based on Eq.(\ref{g18}).

First, due to $O(||\dot{\mathbf{X}}_{P}||)=O(||\dot{U}_P)||)=O(\varepsilon^{\sigma-1})$ and $O(||\frac{\nabla_{\mathbf{X}/U} F_{s0}}{F_{s0}}||)=O(1)$, it's achieved that
\begin{equation}\label{f3}
O(||(\dot{\mathbf{X}}_{P}\cdot \nabla +\dot{U}_P \partial_U)||_0)=O(\varepsilon^{\sigma-1}),
\end{equation}
where the subscript ``0'' denotes the operation on $F_{s0}$. 

Second, $O(||\dot{U}_E\cdot \partial_U||_1 )=O(1)$ and $O(||\dot{\mathbf{X}}_{E} \cdot \nabla||_1)=O(1)$ hold and the subscript ``1'' denotes the operation on the perturbative quantities.  For the latter one, $O(||\frac{\partial_{U} F_{s1}}{F_{s1}}||)=O(1)$ and $O(||\frac{\nabla_{\parallel} F_{s1}}{F_{s1}}||)=O(1)$ are used. Since $O(||\partial_t||_1)=O(1)$ is assumed, 
\begin{equation}\label{f4}
O(||\partial_t +\dot{\mathbf{X}}_{E}\cdot \nabla +\dot{U}_E \partial_U||_1)=O(1)
\end{equation}
holds with respect to the low-frequency perturbations.
Eventually, by combining Eqs.(\ref{f2}),(\ref{f3}) and (\ref{f4}), $O(|F_{s1}|)$ is derived as
$$O(|F_{s1}|)=O(\varepsilon^{\sigma-1}).$$
\end{proof}

Then, the approximation of the distribution on the guiding-center coordinate  can be derived based on the coordinate transform given by Eq.(\ref{ee3})
\begin{eqnarray}\label{qq4}
{{\bar F}_s}\left( {\overline {\bf{Z}} } \right) &=& {F_s}\left( {{\bf{\bar X}}, \bar \mu  + {\varepsilon ^{\sigma  - 1}}g_2^\mu \left( {{\bf{\bar X}},\bar \mu ,\bar \theta } \right)+O(\varepsilon^{2\sigma-2}),\bar U} \right)  \nonumber \\
 &= & {\rm{ }}{F_{s0}}\left( {{\bf{\bar X}},\bar \mu ,\bar U} \right) + \frac{{{\varepsilon ^{\sigma  - 1}}\Psi (\overline {\bf{Z}} )}}{{B(\overline {\bf{Z}} )}}{\partial _{\bar \mu }}{F_{s0}}\left( {{\bf{\bar X}},\bar \mu ,\bar U} \right) + {F_{s1}}\left( {\overline {\bf{Z}} } \right)+O(\varepsilon^{2\sigma-2}),
\end{eqnarray}
whose exact order is $O(\varepsilon^{\sigma-1})$.
The exact full-orbit distribution can be derived by substituting the coordinate transform Eq.(\ref{ee2}) into ${{\bar F}_s}\left( {\overline {\bf{Z}} } \right) $.
According to the transform Eq.(\ref{ee2}), the exact full-orbit distribution is
\begin{equation}\label{e36}
{f_s}\left( {\bf{z}} \right) = {F_s}\left( {{\bf{x}} - \varepsilon{\bm{\rho} _0}\left( {{\bf{x}},{\mu _1},{\theta _1}} \right)+\mathcal{O}_1,{\mu _1} + {\varepsilon ^{\sigma  - 1}}g_2^\mu \left( {{\bf{x}} - \varepsilon {\bm{\rho} _0}({\bf{x}},{\mu _1},{\theta _1}),{\mu _1},{\theta _1}} \right)+\mathcal{O}_2 ,u_1 }\right),
\end{equation}
where $$\mathcal{O}_1=O(\varepsilon^2), \mathcal{O}_2=O(\varepsilon^{2\sigma-2}).$$
Based on the approximation of Eq.(\ref{qq4}), the approximation of ${f_s}\left( {\bf{z}} \right)$ with the exact  order being $O(\varepsilon^{\sigma-1})$ is
\begin{equation}{}\label{qq5}
\begin{split}
 {f_s}\left( {\bf{z}} \right) &= {F_{s}}\left( {{\bf{x}} - \varepsilon {\bm{\rho} _0}\left( {{\bf{x}},{\mu _1},{\theta _1}} \right),{\mu _1},{u_1}} \right)  \\
& + \frac{{{\varepsilon ^{\sigma  - 1}}}}{{B({\bf{x}})}}\left[ {\phi \left( {\bf{x}} \right) - \Phi ({\bf{x}} - \varepsilon{\bm{\rho} _0}\left( {{\bf{x}},{\mu _1},{\theta _1}} \right),{\mu _1})} \right]{\partial _{{\mu _1}}}{F_{s0}}\left( {{\bf{x}} ,{\mu _1},{u_1}} \right) 
+O(\varepsilon^{2})+O(\varepsilon^{2\sigma-2}). 
\end{split}
\end{equation}
Due to $2\le \sigma < 3$,  $O(\varepsilon^{2})$ is lower than $O(\varepsilon^{2\sigma-2})$. 

\subsection{SGM}\label{standardmodel}

By recovering the units, $f_s$ in Eq.(\ref{qq5}) with the uncertain terms ignored becomes
\begin{equation}\label{qq21}
{f_s}\left( {{{\bf{z}}}} \right) \approx {\bar F_s}\left( {{\bf{x}} - \bar{{\bm{\rho}} _0}\left( {{{\bf{z}}}} \right),{\mu _1},{u_1}} \right) + \frac{q_s}{{ B({\bf{x}})}}\left[ {\phi \left( {\bf{x}} \right) - \Phi ({\bf{x}} - \bar{{\bm{\rho}} _0}({\bf{z}}),\mu )} \right]{\partial _{{\mu _1}}}{F_{s0}}\left( {{\bf{x}},{\mu _1},{u_1}} \right),
\end{equation}
with the unit-recovered $\bar{\bm{\rho}}_0$ being
\begin{equation}\label{extr}
{\bar{\bm{\rho}} _0} (\mathbf{x},\mu_1,\theta_1)= \frac{1}{q_s}\sqrt {\frac{{2m_s\mu }}{{B\left( {{{\bf{x}}}} \right)}}} \left( { - {{\bf{e}}_1}\cos \theta_1  + {{\bf{e}}_2}\sin \theta_1 } \right).
\end{equation}

We assume the equilibrium distribution $F_{s0}$ can be decomposed as the product between  the parallel part and the perpendicular part
\begin{equation}\label{qq22}
{F_{s0}}({\bf{x}},{\mu _1},{u_1}) = {n_0}({\bf{x}}){F_{s0\parallel }}\left( {{\bf{x}},{u_1}} \right){F_{s0 \bot }}\left( {{\bf{x}},{\mu _1}} \right),
\end{equation}
with probability conservation being satisfied by
\begin{subequations}\label{qq30}
\begin{eqnarray}
\int {{F_{s0\parallel }}d{u_1}}  &=& 1, \\
\int {{F_{s0 \bot }}\frac{{B({\bf{x}})}}{m_s}d{\mu _1}d{\theta _1} } &=& 1,
\end{eqnarray}
\end{subequations}
where under the equilibrium condition, the metric $B(\mathbf{x})/m_s$ is used.

Then, through the integral ${n_s}\left( {{\bf{x}},t} \right) = \int {{f_s}\left( {\bf{z}} \right)\frac{{B({\bf{x}})}}{{{m_s}}}d{\mu _1}d{u_1}d{\theta _1}} $, the density can be assembled as
\begin{equation}\label{qq6}
{n_s}\left( {{\bf{x}}} \right) = {n_{s0}}({\bf{x}}) + \frac{{q_s ({\bf{x}})}}{ B(\bf{x})}\left[ {\Lambda ({\bf{x}}) - \tilde \phi '({\bf{x}})} \right] + {n_{s1}}\left( {{\bf{x}}} \right),
\end{equation}
with
\begin{subequations}\label{qq7}
\begin{eqnarray}
{n_{s1}}\left( {{\bf{x}},t} \right) & = & \int {{{\bar F}_{s1}}\left( {{\bf{x}} - \bar{{\bm{\rho}} _0}\left( \mathbf{z}\right),{\mu _1},{u_1}} \right)\frac{B\left( {\bf{x}} \right)}{m_s}d{\mu _1}d{u_1}d{\theta _1}}, \\
\Lambda ({\bf{x}}) &=& \frac{{B({\bf{x}})}}{{{m_s}}}\int { \phi(\mathbf{x}){\partial _{{\mu _1}}}{F_{s0 }}({\bf{x}},{\mu _1},u_1)d{\mu _1} d u_1 d{\theta _1}}, \\
\tilde \phi' \left( {\bf{x}} \right) &=& \frac{ B(\mathbf{x})}{{{m_s}}}\int {\Phi \left( {{\bf{x}} - \bar{{\bm{\rho}} _0}(\mathbf{z}),{\mu _1}} \right){\partial _{{\mu _1}}}{F_{s0  }}\left( {{\bf{x}},{\mu _1}},u_1 \right)d{\mu _1}}du_1 d\theta_1,
\end{eqnarray}
\end{subequations}
Here, the metric $\eta_2(\mathbf{z})$ equaling $B(\mathbf{x})/m_s$ of the phase space is used.
$\tilde \phi' \left( {\bf{x}} \right)$ is the so-called double-gyroaverage term. The term of $\Phi \left( {{\bf{x}} - \bar{{\bm{\rho}} _0},{\mu _1}} \right)$ can be derived from Eq(\ref{ee9}).

If we consider a plasma only including protons and electrons and the electrons obey the adiabatic distribution,
\begin{equation}\label{qq26}
{n_e}({\bf{x}}) = {n_0}({\bf{x}}) + \frac{{e{n_0}({\bf{x}})}}{{{T_e}}}\phi ({\bf{x}}),
\end{equation}
QNE of this plasma is
\begin{equation}\label{qq27}
-\frac{{e\Lambda ({\bf{x}})}}{{B({\bf{x}})}} + \frac{{e\tilde \phi ' ({\bf{x}})}}{{B({\bf{x}})}} + \frac{e}{{{T_e}}}\phi ({\bf{x}}) - \frac{{{n_{s1}}\left( {\bf{x}} \right)}}{{{n_{s0}}({\bf{x}})}} = 0.
\end{equation}

The following equilibrium distribution which the magnetic moment satisfies is chosen in this paper
\begin{equation}\label{perp5}
{F_{0 \bot }} = \frac{{{m_s}}}{{2\pi {T_i}}}\exp \left( -{\frac{{\mu B}}{{{T_i}}}} \right).
\end{equation}
By substituting $F_{0\perp}$, $\Lambda$ can be derived as
$$\Lambda(\mathbf{x})=-\frac{B\phi(\mathbf{x})}{T_i(\mathbf{x})}, $$
and QNE becomes 
\begin{equation}\label{perp7}
\frac{{e\phi ({\bf{x}})}}{{{T_i}}} + \frac{{e\tilde \phi '({\bf{x}})}}{{B}} + \frac{e\phi ({\bf{x}})}{{{T_e}}} - \frac{{{n_{s1}}}}{{{n_{s0}}}} = 0,
\end{equation}
where 
$$\tilde \phi'\left( {\bf{x}} \right)=-\frac{ B^2}{{{T_i m_s}}}\int {\Phi \left( {{\bf{x}} - \bar{{\bm{\rho}} _0}(\mathbf{z}),{\mu _1}} \right){F_{s0  }}\left( {{\bf{x}},{\mu _1}},u_1 \right)d{\mu _1}}du_1 d\theta_1,$$
and  $n_{s1}(\mathbf{x})$ are given by Eq.(\ref{qq7}) with $F_{0\perp}$ in Eq.(\ref{perp5}).


\section{The Lie transform perturbative method}\label{lie}

This method was given in Ref.\cite{Cary1983} and it begins with the following autonomous differential equations
\begin{equation}\label{a117}
\frac{{\partial Y_f^i }}{{\partial \epsilon }}\left( {{\bf{y}},\epsilon } \right) = {g_1^i }\left( {{{\bf{Y}}_f}\left( {{\bf{y}},\epsilon } \right)} \right),
\end{equation}
\begin{equation}\label{a118}
\frac{{d{\bf{y}}}}{{d\epsilon }} = 0,
\end{equation}
where $\mathbf{Y}=\mathbf{Y}_f(\mathbf{y},\epsilon)$ is the new coordinates, $\mathbf{y}$ is the old coordinates, and $\epsilon$ is an independent variable denoting the small parameter of amplitude of perturbation.
Eqs.(\ref{a117}) and (\ref{a118}) lead to the solution
\begin{equation}\label{c1}
{\bf{y}} = \exp \left(-{{\epsilon g^i_1}{\partial _{{Y_i}}}} \right){\bf{Y}},
\end{equation}
where the Einstein summation is used.
For a differential 1-form written as $\gamma(\bf{z})$, which doesn't depend on $\epsilon$ in the coordinate frame of $\bf{z}$, coordinate transform iy Eq.(\ref{c1}) induces a pullback transform of $\gamma$ as
\begin{equation}\label{c2}
{\Gamma _i }\left(\mathbf{ Y} \right) = {\left[ {\exp \left( { - \varepsilon {L_{1}}} \right)\gamma } \right]_i }\left( \mathbf{Y} \right) + \frac{{\partial S\left( \mathbf{Y} \right)}}{{\partial {Y^i }}}dY^i.
\end{equation}
where $S(\mathbf{Y})$ is a gauge function and the $i$ component of $L_1\gamma$ is defined as ${\left( {{L_1}\gamma } \right)_i } = g_1^j\left( {{\partial _j}{\gamma _i } - {\partial _i }{\gamma _j}} \right)$.

When the differential 1-form explicitly depends on the perturbation and can be written as
$\gamma(\mathbf{y},\varepsilon)  = {\gamma _0}(\mathbf{y}) + \epsilon {\gamma _1}(\mathbf{y}) + \epsilon ^{2}{\gamma _2}(\mathbf{y}) +  \cdots$,
Ref.\cite{Cary1983} generalizes Eq.(\ref{c2}) to be a composition of individual Lie transforms $T =  \cdots {T_3}{T_2}{T_1}$ with
\begin{equation}\label{c4}
{T_n} = \exp \left( { - \epsilon^n {L_{n}}} \right),
\end{equation}
to get the new 1-form
\begin{equation}\label{f1}
\Gamma  = T\gamma  + dS,
\end{equation}
which can be expanded by the order of $\epsilon$
\begin{equation}\label{c5}
{\Gamma _0} = {\gamma _0},
\end{equation}
\begin{equation}\label{c6}
{\Gamma _1} = d{S_1} - {L_1}{\gamma _0} + {\gamma _1},
\end{equation}
\begin{equation}\label{c7}
{\Gamma _2} = d{S_2} - {L_2}{\gamma _0} + {\gamma _2} - {L_1}{\gamma _1} + \frac{1}{2}L_1^2{\gamma _0},
\end{equation}
\begin{equation}\label{a33}
\cdots   \nonumber \\
\end{equation}
These expanding formulas can be written in a general form
\begin{equation}\label{c8}
\Gamma_n = d S_n - L_n \gamma_0 + C_n.
\end{equation}
By requiring $\Gamma_{ni}=0,i\in(1,\cdots,2N)$, the $n$th order generators are
\begin{equation}\label{c9}
g_n^j = \left( {\frac{{\partial {S_n}}}{{\partial {y^i}}} + {C_{ni}}} \right)J_0^{ij},
\end{equation}
where $J_0^{ij}$ is Poisson tensor.
And correspondingly, the $n$th order gauge function can be solved as
\begin{equation}\label{c10}
V_0^i \frac{{\partial {S_n}}}{{\partial {y^i}}} = \frac{{\partial {S_n}}}{{\partial {y^0}}} + V_0^i\frac{{\partial {S_n}}}{{\partial {y^i}}} = {\Gamma _{n0}} - {C_{n i }}V_0^i
\end{equation}
with
\begin{equation}\label{c11}
V_0^ i= J_0^{ij}\left( {\frac{{\partial {\gamma _{0j}}}}{{\partial {y^0}}} - \frac{{\partial {\gamma _{00}}}}{{\partial {y^j}}}} \right)
\end{equation}
To avoid the secularity of $S_n$, usually $\Gamma_{n0}$ is chosen to be
\begin{equation}\label{c12}
{\Gamma _{n0}} = \left[\kern-0.15em\left[ {V_0^i {C_{ni }}}
 \right]\kern-0.15em\right],
\end{equation}
where $\left[\kern-0.15em\left[  \cdots
 \right]\kern-0.15em\right]$ means average over the fast variable.

\section{The non-zero components of the Lie derivatives on $\Gamma_0$ in Eq.(\ref{g8})}\label{app1}

The formula of the Lie derivative of the generators on the differential 1-form $\gamma=\gamma_a dz^a$ is given as
\begin{equation}\label{c13}
{L_{\bf{g}}}\gamma  = \left( {{g^a}{\omega _{ab}} + {\partial _b}\left( {{g^a}{\gamma _a}} \right)} \right)d{z^b},
\end{equation}
where $\gamma_a$ is the component corresponding to $z^a$. $\omega$ is the Poisson bracket defined as ${\omega _{ab}} = {\partial _{{z^a}}}{\gamma _b} - {\partial _{{z^b}}}{\gamma _a}$. The part ${\partial _b}\left( {{g^a}{\gamma _a}} \right)d{z^b}$ in Eq.(\ref{c13}) is a full differential term and can be treated as a gauge term.
In this paper, the generator vector $\mathbf{g}$ is given as ${\bf{g}} \equiv ({{\bf{g}}^{\bf{x}}},{g^\mu },{g^U},{g^\theta })$ with ${{\bf{g}}^{\bf{x}}} = ({g^1},{g^2},{g^3})$ for the spatial space. $g^\mu$,$g^U$ and $g^\theta$ are for the dimensions of $\mu,U,\theta$, respectively. And the specific $\gamma$ is given by $\Gamma_0$ in Eq.(\ref{g8}).
The nonzero components of the Lie derivative on $\Gamma_0$ in Eq.(\ref{g8}) are given below.
\begin{subequations}
\begin{eqnarray}
g_{}^i{\omega _{0ij}}d{X^j}
&=& \left( {{\bf{B}} + \varepsilon U\nabla  \times b} \right) \times {\bf{g}}^\mathbf{x} \cdot d{\bf{X}}, \\
g_{}^U{\omega _{0Ui}}d{X^i}
&=& \varepsilon g_{}^U{\bf{b}}\cdot d{\bf{X}}, \\
g_{}^i{\omega _{0iU}}dU
&=&  - \varepsilon \left( {{\bf{g}}^\mathbf{x}\cdot{\bf{b}}} \right)dU, \\
g_{}^\mu {\omega _{0\mu \theta }}d\theta
&=& {\varepsilon ^2}g_{}^\mu d\theta,  \\
g_{}^\theta {\omega _{0\theta \mu }}d\mu
&=&  - {\varepsilon ^2}g_{}^\theta d\mu,  \\
g_{}^j{\omega _{0jt}}dt
&=&  - \varepsilon \mu {\bf{g}}^\mathbf{x}\cdot\nabla B\left( {\bf{X}} \right)dt,  \\
g_{}^\mu {\omega _{0\mu t}}dt
&=&  - \varepsilon B\left( {\bf{X}} \right)g_{}^\mu dt,  \\
g_{}^U{\omega _{0Ut}}dt
&=&  - \varepsilon Ug_{}^Udt.
\end{eqnarray}
\end{subequations}

\section{The expansion of the function over the small parameters}\label{expand}

 We first consider a function of the form $f(x+\epsilon g(x))$ depending on one scalar argument  and a small parameter. What we are interested in is its expansion over $\epsilon$. The derivative of $f(x+\epsilon g(x))$ over $\epsilon$ at $\epsilon=0$ is derived as follows
\begin{equation}\label{q6}
\left.{ {d_\epsilon }f(z)} \right|_{\epsilon  = 0}= {\left. {{d_\epsilon }z{\partial _z}f(z)} \right|_{\epsilon  = 0}} = g(x){\partial _x}f\left( x \right),z \equiv x + \epsilon g(x),
\end{equation}
where $d_\epsilon\equiv d/d\epsilon$. The second order derivative of $f(z)$ over $\epsilon$ is 
\begin{equation} \nonumber
\begin{split}
d_\epsilon(d_\epsilon f(z))&=d_\epsilon (d_\epsilon z \partial_z f(z))=d_\epsilon (d_\epsilon z) \partial_z f(z)+d_\epsilon z d_\epsilon \partial_z f(z) \\
&=d_\epsilon g(x)+d_\epsilon z d_\epsilon z \partial_z \partial_z f(z)=g^2(x)\partial^2_z f(z) .
\end{split}
\end{equation}
Then, the second order derivative of $f(z)$ over $\epsilon$ at $\epsilon=0$ is
$${\left.d_\epsilon^2f(x+\epsilon g(x))\right |_{\epsilon  = 0}}=g^2(x)\partial^2_x f(x).$$
It's easy to derive that the $n$-th derivative of $f(x+\epsilon g(x))$ over $\epsilon$ at $\epsilon=0$ is
\begin{equation}\label{q1}
{\left. {d _\epsilon ^nf(x + \epsilon g(x))} \right|_{\epsilon  = 0}} = {g^n}(x)\partial _x^nf(x).
\end{equation}
Then, the Taylor expansion of $f(x+\epsilon g(x))$ over $\epsilon$ is
\begin{equation}\label{q2}
f(x + \epsilon g(x)) = \sum\limits_{n \ge 0} {\frac{{{\epsilon ^n}}}{{n!}}{g^n}(x)\partial _x^nf(x)}.
\end{equation}
If there are two independent small parameters $\{\epsilon,\epsilon_1\}$, and the argument of $f$ is like $x+\epsilon g(x) +\epsilon_1 g_1(x)$, the expanding of $f(x+\epsilon g(x) +\epsilon_1 g_1(x))$ over $\{\epsilon,\epsilon_1\}$ is
\begin{equation}\label{q3}
f(x + \epsilon g(x) + {\epsilon _1}{g_1}(x)) = \sum\limits_{n,{n_1} \ge 0} {\frac{{{\epsilon ^n}\epsilon _1^n}}{{n!{n_1}!}}{g^n}(x)g_1^{{n_1}}(x)\partial _x^{n + {n_1}}f(x)}.
\end{equation}

Now we change $g(x)$ to be a multiple variable vector $\mathbf{g}(\mathbf{x})$. In Cartesian coordinate frame, $\mathbf{g}(\mathbf{x})\cdot \nabla$ can be written as
$$\mathbf{g}(\mathbf{x})\cdot \nabla=\sum \limits_i g_i(\mathbf{x})\partial'_{x_i}.$$
where $'$  means $\partial'_{x_i}$ doesn't operate on any $g_i(\mathbf{x})$. Then,  Eq.(\ref{q2})and (\ref{q3}) are respectively changed to be
\begin{equation}\label{q4}
f({\bf{x}} + \epsilon {\bf{g}}({\bf{x}})) = \sum\limits_{n \ge 0} {\frac{{{\epsilon ^n}}}{{n!}}{{\left( \sum \limits_i g_i(\mathbf{x})\partial'_{x_i}\right)}^n}f({\bf{x}})} .
\end{equation}
\begin{equation}\label{q5}
f({\bf{x}} + \epsilon {\bf{g}}({\bf{x}}) + {\epsilon _1}{{\bf{g}}_1}({\bf{x}})) = \sum\limits_{n,{n_1} \ge 0} {\frac{{{\epsilon ^n}\epsilon _1^n}}{{n!{n_1}!}}{{\left( \sum \limits_i g_{i}(\mathbf{x})\partial'_{x_i} \right)}^n}{{\left( \sum \limits_i g_{1i}(\mathbf{x})\partial'_{x_i} \right)}^{{n_1}}}f({\bf{x}})}
\end{equation}
In Eq.(\ref{q4}) and (\ref{q5}), the superscript $'$ means that the derivative $\partial_{\mathbf{x}}$ only acts upon $f(\bf{x})$.

When the argument of $f$ is of the form $x+\epsilon g(x)+\epsilon \epsilon_1 g_1(x)$, the general derivatives of $f$ such as ${\left. {\partial _\epsilon ^n\partial _{{\epsilon _1}}^{{n_1}}f\left( {x + \epsilon g(x) + \epsilon {\epsilon _1}{g_1}(x)} \right)} \right|_{\epsilon  = 0,{\epsilon _1} = 0}}$ doesn't have an uniform formula like that given by Eq.(\ref{q3}). Fortunately, we don't need higher order composite derivatives in this paper.

\emph{\textbf{Remark}}: As Eq.(\ref{q3}) shows, the expanding of $f$ over several small parameters $\epsilon_i$s doesn't contain the mutual derivative between $g_i(x)\partial_x$ and $g_j(x)\partial_x$ such as $g_i(x)\partial_x g_j(x)$.

\end{appendices}

\bibliographystyle{plain}
\bibliography{hybrid_coordinate.bib}

\end{document}